\documentclass
[preprint,byrevtex,nofootinbib,10pt,balancelastpage,titlepage,bibnotes,twocolumn,pre]{revtex4}%
\usepackage{amsfonts}
\usepackage{amsmath}
\usepackage{amssymb}
\usepackage{graphicx}%
\providecommand{\U}[1]{\protect\rule{.1in}{.1in}}
\newtheorem{theorem}{Theorem}

\newtheorem{claim}[theorem]{Claim}
\newtheorem{conclusion}[theorem]{Conclusion}
\newtheorem{condition}[theorem]{Condition}

\newtheorem{definition}[theorem]{Definition}

\newtheorem{proposition}[theorem]{Proposition}
\newtheorem{remark}[theorem]{Remark}

\newenvironment{proof}[1][Proof]{\noindent\textbf{#1.} }{\ \rule{0.5em}{0.5em}}
\ifx\pdfoutput\relax\let\pdfoutput=\undefined\fi
\newcount\msipdfoutput
\ifx\pdfoutput\undefined\else
\ifcase\pdfoutput\else
\ifx\paperwidth\undefined\else
\ifdim\paperheight=0pt\relax\else\pdfpageheight\paperheight\fi
\ifdim\paperwidth=0pt\relax\else\pdfpagewidth\paperwidth\fi
\fi\fi\fi
\begin{document}
\preprint{UATP/1804}
\title{Correcting the Mistaken Identification of Nonequilibrium Microscopic Work}
\author{P.D. Gujrati}
\email{pdg@uakron.edu}
\affiliation{Department of Physics, Department of Polymer Science, The University of Akron,
Akron, OH 44325}

\begin{abstract}
Nonequilibrium work-Hamiltonian connection at the level of a microstate plays
a central role in diverse modern branches of statistical thermodynamics
(fluctuation theorems, quantum thermodynamics, stochastic thermodynamics,
etc.). The energy change $dE_{k}$ for the $k$th microstate is \emph{always}
but \emph{erroneously} equated with the external work $d\widetilde{W_{k}}$
done \emph{on} the microstate, whereas the correct identification is
$dE_{k}=-dW_{k}$, $dW_{k}$ being the generalized work done \emph{by }the
microstate; the average of $d_{\text{i}}W_{k}=d\widetilde{W_{k}}+dW_{k}$
represents the dissipated work. We show that $d_{\text{i}}W_{k}$ is ubiquitous
even in a purely mechanical system, let alone a nonequilibrium system.\ As
$d_{\text{i}}W_{k}$ is not included in $d\widetilde{W_{k}}$, the current
identification $dE_{k}=d\widetilde{W_{k}}$ does not account for
irreversibility such as in the Jarzynski equality (JE). Using $dW_{k}$ to
account for irreversibility, we obtain a new work relation that works even for
free expansion, where the JE fails. In the new work relation, $\Delta
E_{k}=-\Delta W_{k}$ depends only on the energies of the initial and final
states and not on the actual process. This makes the new relation very
different from the JE. Thus, the correction has far-reaching consequences and
requires reassessment of current applications in theoretical physics.

\end{abstract}
\date{January 31, 2017}
\maketitle

\section{Introduction}

\subsection{Motivation}

In this work, we are mainly interested in an interacting system $\Sigma$ in
the presence of a medium $\widetilde{\Sigma}$ as shown in Fig.
\ref{Fig_System}; the latter will be taken to be a collection of two parts: a
work source $\widetilde{\Sigma}_{\text{w}}$ and a heat source $\widetilde
{\Sigma}_{\text{h}}$, both of which can interact with the system $\Sigma$
directly but not with each other. We will continue to use $\widetilde{\Sigma}$
to refer to both of them together. The collection $\Sigma_{0}=\Sigma
\cup\widetilde{\Sigma}$ forms an isolated system, which we assume to be
stationary, with all of of its observables such as its energy $E_{0}$, its
volume $V_{0}$, its number of particles (we assume a single species) $N_{0}$
constant in time, even if there may be internal changes going on when its
parts are out of equilibrium. These internal changes require some internal
variables for their description \cite{deGroot,Prigogine,Maugin}. In the
following, we will assume $\widetilde{\Sigma}$ to be always in equilibrium
(which requires it to be extremely large compared to $\Sigma$) so that any
irreversibility going on within $\Sigma_{0}$ is ascribed to $\Sigma$ alone,
and is caused by processes such as dissipation due to viscosity, internal
inhomogeneities, etc. that are internal to the system. Moreover, we assume
additivity of volume, a weak interaction between, and quasi-independence of,
$\Sigma$ and $\widetilde{\Sigma}$; the last two conditions ensure that the
energies and entropies are additive
\cite{Gujrati-I,Gujrati-II,Gujrati-Entropy2,Gujrati-Entropy1,Gujrati-III} but
also impose some restriction on the size of $\Sigma$ in that it cannot be too
small. In particular, the size should be at least as big as the correlation
length for quasi-independence as discussed earlier
\cite{Gujrati-II,Gujrati-Entropy2,Gujrati-Entropy1}. In this work, we will
assume that all required conditions necessary for the above-mentioned
additivity are met.

All quantities pertaining to $\Sigma$\ carry no suffix; those for
$\widetilde{\Sigma}$ carry a tilde, while those of $\Sigma_{0}$ carry a suffix
$0$. As $\widetilde{\Sigma}$ is always in equilibrium, all its fields have the
same value as for $\Sigma_{0}$ so that its temperature, pressure, etc. are
shown as $T_{0},P_{0}$, etc. as seen in Fig. \ref{Fig_System}. To simplify our
discussion, we are going to mostly restrict our discussion to two parameters
in the Hamiltonian; however, later in Secs. \ref{Sec-Theory}%
-\ref{Sec-WorkTheorem}, we will consider a general set $\mathbf{W}$ of work
parameters. In equilibrium (EQ), the state of $\Sigma$ is described by two
observables\textbf{ }$E$ and $V$ of which $V$ appears as a parameter in its
Hamiltonian $\mathcal{H}(\mathbf{z},V)$; here, $\mathbf{z}$ is the collection
of coordinates and momenta of the $N$ particles. Throughout this work, $N$ is
kept fixed so we will not exhibit it unless necessary. We can manipulate $V$
from the outside through $\widetilde{\Sigma}_{\text{w}}$, which exerts an
external "force" such as the pressure to do some "work." For this work, $V$ is
a symbolic representation of an externally controlled parameter, which is
commonly denoted by $\lambda$ in the current literature and whose nature is
determined by the experimental setup. Thus, $V$ here should be considered a
parameter in a wider sense than just the volume and $P$ its conjugate force;
see Fig. (\ref{Fig_Piston-Spring}).

\begin{definition}
\label{Def-State}At the microscopic level, the state of the system is
specified by the set of \emph{microstates} $\left\{  \mathfrak{m}_{k}\right\}
$, their energy set $\left\{  E_{k}\right\}  $ and their \emph{probability}
set $\left\{  p_{k}\right\}  $. For the same set $\left\{  \mathfrak{m}%
_{k},E_{k}\right\}  $, different choices of $\left\{  p_{k}\right\}  $
describes different states, one of which corresponding to $\left\{
p_{k\text{eq}}\right\}  $ uniquely specifies an EQ state; all other states are
called nonequilibrium (NEQ) states.
%TCIMACRO{\FRAME{ftbpFU}{3.2145in}{2.6238in}{0pt}{\Qcb{An isolated system
%$\Sigma_{0}$ consisting of the system $\Sigma$ in a surrounding medium
%$\widetilde{\Sigma}$. The medium and the system are characterized by their
%fields $T_{0},P_{0},...$ and $T(t),P(t),...$, respectively, which are
%different when the two are out of equilibrium. }}{\Qlb{Fig_System}%
%}{Interacting-System-Exchange0.eps}{\special{ language "Scientific Word";
%type "GRAPHIC";  maintain-aspect-ratio TRUE;  display "USEDEF";
%valid_file "F";  width 3.2145in;  height 2.6238in;  depth 0pt;
%original-width 4.9018in;  original-height 3.9946in;  cropleft "0";
%filename 'Interacting-System-Exchange0.eps';file-properties "XNPEU";}} }%
%BeginExpansion
\begin{figure}
[ptb]
\begin{center}
\includegraphics[
height=2.6238in,
width=3.2145in
]%
{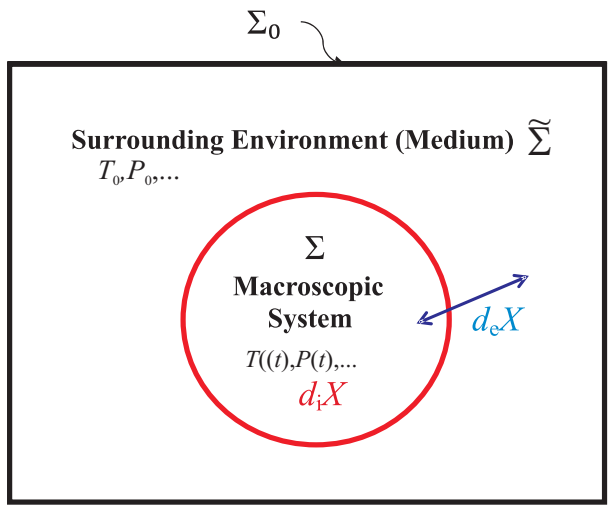}%
\caption{An isolated system $\Sigma_{0}$ consisting of the system $\Sigma$ in
a surrounding medium $\widetilde{\Sigma}$. The medium and the system are
characterized by their fields $T_{0},P_{0},...$ and $T(t),P(t),...$,
respectively, which are different when the two are out of equilibrium. }%
\label{Fig_System}%
\end{center}
\end{figure}
%EndExpansion

\end{definition}

The EQ entropy is a state function $S(E,V)$. Away from EQ, we also need an
additional set of independent extensive internal variables to specify the
state \cite{deGroot,Prigogine,Maugin}; we take a single variable $\xi$ for our
study to simplify the presentation \cite{Note-xi-dependence}. In this case,
$V$ and $\xi$ appear as parameters in the Hamiltonian $\mathcal{H}%
(\mathbf{z},V,\xi)$. However, the discussion can be readily extended to many
internal variables. The force conjugate to $\xi$\ is known as the
\emph{affinity }and conventionally denoted by $A$
\cite{deGroot,Prigogine,Maugin}. A NEQ state for which its entropy is a state
function $S(E,V,\xi)$ is said to be an \emph{internal equilibrium} (IEQ) state
\cite{Gujrati-I,Gujrati-II}. We will only consider IEQ states in this work for
simplicity. It is well known that the energy $E=\mathcal{H}_{V,\xi}$ does not
change at fixed $V$ and $\xi$ as $\mathbf{z}$ evolves in accordance with
Hamilton's equations of motion. Therefore, $\mathbf{z}$ plays no role in the
performance of work so we will usually not exhibit it unless clarity is needed.

Focus on the Hamiltonian $\mathcal{H}(V,\xi)$ provides us with the advantage
of obtaining a microscopic description of NEQ thermodynamics, where a central
role is played by the microscopic work-Hamiltonian relation at the level of
microstates\emph{ }$\left\{  \mathfrak{m}_{k}\right\}  $ of a NEQ system
$\Sigma$. The work relation is a key ingredient in developing a microscopic
NEQ statistical thermodynamics, where the emphasis is in the application of
the first law in diverse branches including but not limited to
\emph{nonequilibrium work theorems,}
\cite{Bochkov,Jarzynski,Crooks,Pitaevskii} \emph{stochastic}
\emph{thermodynamics,} \cite{Sekimoto,Seifert} and \emph{quantum
thermodynamics,} \cite{Lebowitz,Alicki} to name a few. Once microscopic work
is identified, microscopic heat can also be identified by invoking the first
law. This makes identifying work of primary importance in NEQ thermodynamics.
Unfortunately, this endeavor has given rise to a controversy about the actual
meaning of NEQ work, which apparently is far from settled
\cite{Cohen,Jarzynski-Cohen,Sung,Gross,Jarzynski-Gross,Peliti,Rubi,Jarzynski-Rubi,Rubi-Jarzynski,Peliti-Rubi,Rubi-Peliti,Pitaevskii,Bochkov}%
.

\subsection{Mistaken Identity\label{Sec-MistakenIdentity}}

The above controversy is distinct from the confusion about the meaning of work
and heat in classical nonequilibrium thermodynamics \cite{Fermi,Kestin}
involving a \emph{system-intrinsic }(SI)\emph{ }or\emph{ medium-intrinsic}
(MI)\ description, the latter only recently having been clarified
\cite{Gujrati-Heat-Work0,Gujrati-Heat-Work,Gujrati-I,Gujrati-II,Gujrati-III,Gujrati-Entropy1,Gujrati-Entropy2}%
; see Sec. \ref{Sec-NewConcepts} for more details and clarification.

\begin{definition}
\label{Def-SI-MI}By a SI-description (MI-description), we mean a description
that uses quantities associated with $\Sigma$ only (with $\widetilde{\Sigma}$
only); see for example Eqs.(\ref{SI-Work}) for $dW$\ and (\ref{MI-Work}) for
$d\widetilde{W}$.
\end{definition}

The SI-work $dW$ is the work done by the system and the MI-work $d\widetilde
{W}$ is the work done by the medium. At the level of microstates, one is faced
with the fact that both of these works $dW_{k}$\ and $d\widetilde{W}_{k}%
$\ cannot be related to the change $dE_{k}$ in the microstate energy of
$\mathfrak{m}_{k}$. Indeed, as pointed out recently
\cite{Gujrati-GeneralizedWork}, the identification of microscopic work has
proven erroneous in that the external work $d\widetilde{W}_{k}$, an MI
quantity for $\mathfrak{m}_{k}$, is identified with the SI energy change
$dE_{k}$ for $\mathfrak{m}_{k}$, whereas in a NEQ state, the thermodynamics of
the medium and the system are very different. As pointed above, it is $\Sigma
$\ that experiences dissipation \cite{Fermi,Kestin,Woods}, which is partly
responsible for the irreversible entropy generation $d_{\text{i}}S$ in the
system. This dissipation is not captured by $d\widetilde{W}_{k}$ in its
current formulation as we will see in this work. This is related to a
surprising fact hitherto unappreciated that at the microscopic level,
$d\widetilde{W}_{k}$ and $dE_{k}$ always differ in a thermodynamic system,
even when dealing with a reversible process as shown by a thermodynamic
example in Sec. \ref{Sec-ForceImbalance-Thermodynamic}. In fact, they differ
even for a purely mechanical system as shown by the mechanical example in Sec.
\ref{Sec-Example-MechSystem}.

\textbf{Jarzynski Equality}: As we will show here (this has been briefly
reported earlier \cite{Gujrati-GeneralizedWork}), not recognizing this subtle
fact has given rise to the "mistaken identity"
\begin{equation}
d\widetilde{W}_{k}\overset{?}{=}dE_{k}, \label{MistakenIdentity}%
\end{equation}
between an MI-work and SI-energy change in the current literature; see for
example Jarzynski \cite{Jarzynski} who exploits the mistaken identity to prove
by now the famous Jarzynski equality (JE) discussed below in Sec.
\ref{Sec-JarzynskiProcess}.

Following Prigogine \cite{Prigogine}, see also Fig. (\ref{Fig_System}) and
Sec. \ref{Sec-Notation} for clarification, we partition
\begin{equation}
dE_{k}=d_{\text{e}}E_{k}+d_{\text{i}}E_{k}, \label{E_k-partition}%
\end{equation}
in which $d_{\text{e}}E_{k}\equiv d\widetilde{W}_{k}$ represents the change
caused by the exchange with $\widetilde{\Sigma}$ through external work and
$d_{\text{i}}E_{k}$ represents the change caused by internal processes going
on within $\Sigma$. As the work $dW_{k}$\ is the work done by $\mathfrak{m}%
_{k}$ at the expense of its energy loss, we have $dW_{k}=-dE_{k}$ as we will
establish here. We will show\ by examples and arguments that $d_{\text{i}%
}E_{k}\neq0$ is \emph{ubiquitous} not only in purely mechanical systems but
also in a reversible process and is \emph{necessary for dissipation}. The main
thesis of this work to be established here is the following

\begin{proposition}
\label{Proposition-Goal}It is the contribution $d_{\text{i}}E_{k}$ that is
necessary but not sufficient to describe dissipation but is not captured by
$d\widetilde{W}_{k}=d_{\text{e}}E_{k}$. Therefore, assuming $d\widetilde
{W}_{k}=dE_{k}$ is to ignore dissipation completely regardless of the time
dependence of the work protocol. The SI-work relation $dW_{k}=-dE_{k}$, which
will be established here, contains dissipation and will provide the required
NEQ work relation.
\end{proposition}

The same conclusion also applies to the accumulation $\Delta_{\text{i}}%
E_{k}\neq0$ along a trajectory, which is the path $\gamma$ taken by the $k$th
microstate $\mathfrak{m}_{k}$ during a NEQ process $\mathcal{P}_{0}$ over a
time interval $\left(  0,\tau\right)  $. The energy changes $dE_{k}%
,d_{\text{e}}E_{k}$ and $d_{\text{i}}E_{k}$ occur during a segment of
$\mathcal{P}_{0}$ between $t$\ and $t+dt$, while $\Delta E_{k},\Delta
_{\text{e}}E_{k}$ and $\Delta_{\text{i}}E_{k}$ are cumulative changes over the
entire trajectory $\gamma$; see Sec.\ref{Sec-MicrostateThermodynamics} for
more details.

\textbf{Quantum and Stochastic Thermodynamics}: In a seminal paper laying down
the foundation of quantum thermodynamics, Alicki \cite[cf. Eq. (2.4)
there]{Alicki} identifies the average work $d\widetilde{W}$, the average of
$d\widetilde{W}_{k}$, done by external forces with the average of $dE_{k}$ in
accordance with Eq. (\ref{MistakenIdentity}). In a paper of similar importance
for stochastic thermodynamics, Sekimoto \cite[cf. Eqs. (2.7)-(2.9)
there]{Sekimoto} also identifies the change $dE_{k}$ in the energy $E_{k}$ by
work variables with the work $d\widetilde{W}_{k}$ done by the external system.

\textbf{What is Accomplished}: The recent report
\cite{Gujrati-GeneralizedWork} was very compact. This work is designed to
elaborate on what was reported, to provide the missing details and more.
Throughout the work, we include an internal varaiable as a work variable. The
main emphasis here will be to demonstrate the ubiquitous nature of
$d_{\text{i}}E_{k}$, whose existence has not been previously appreciated. Not
recognizing its existence has resulted in the above mistaken identification.
It is the force imbalance between the internal and external forces that
generates $d_{\text{i}}E_{k}$ and is very common in almost processes, even a
reversible one as we will demonstrate. This is the most important outcome of
the our approach, which emphasizes the importance of SI-quantities (such as
$dE_{k}=-dW_{k}$) that are very different from the MI-quantities (such as
$d_{\text{e}}E_{k}=d\widetilde{W}_{k}$) in any\ process, even a reversible
one. The use of generalized work as isentropic change $dW=-dE_{S}$ allows us
to calculate microscopic work $dW_{k}$, which changes $dE_{k\text{ }}$but not
$p_{k}$.\ On the other hand, the generalized heat $dQ_{k}$ does not change
$dE_{k}$ but changes $p_{k}$. This proves to be a great simplification and
allows us to treat $dW_{k}$ and $dQ_{k}$ as purely a mechanical and a
stochastic concept, respectively. In addition, as $\mathfrak{m}_{k}$ uniquely
determines $E_{k}$ for a fixed work set $\mathbf{W}$, we find that $dE_{k}$ is
uniquely determined by $d\mathbf{W}$ and not by the nature of the work
process. This provides a simplification in evaluating the cumulative change
$\Delta W_{k}$ and allows us to prove a new work theorem, both for the
exclusive and the inclusive Hamiltonian. This work theorem is very general and
we apply it successfully to the free expansion.

\subsection{Layout}

The layout of the paper is the following. In the next section, we introduce
the extension of Prigogine's notation as it may be unfamiliar to many readers.
In Sec. \ref{Sec-NewConcepts}, three important concepts of NEQ quantities that
will play a central role in this work: the concepts of SI- and MI-quantities,
and that of microwork. In Sec. \ref{Sec-FirstLaw-ThermodynamicForces}, we give
two alternate versions of the first law using the MI- and SI-quantities, and
introduce thermodynamic forces and their role in dissipation that does not
seem to have been appreciated in the field. A new concept of force imbalance
is also introduced that plays an important role in microscopic thermodynamics
but has been hitherto overlooked. It is one of the central new quantities that
is ubiquitous and needs to be incorporated in any process, EQ or NEQ. In Sec.
\ref{Sec-MicrostateThermodynamics}, we introduce a new approach to investigate
NEQ thermodynamics at the microscopic level by using microstates. Here, we
introduce a different way of looking at work and heat, which we call
generalized work and generalized heat, that has only recently been developed
but proves very useful for a microscopic understanding of statistical
thermodynamics. The generalized work is isentropic and the generalized heat
satisfies the Clausius equality $dQ=TdS$ even for an irreversible process.
This makes generalized work a purely mechanical quantity. The stochasicity
arises from heat but not from work. This results in a clear division between
work and heat, making them independent quantities as we explain here. In the
next section, we introduce exclusive and inclusive Hamiltonians. It is the
latter Hamiltonian that is used by Jarzynski, but we find basically no formal
difference between the two. In Sec. \ref{Sec-Examples}, we consider two
different models, a mechanical and a thermodynamic, to explain various
different kinds of NEQ work and heat. In Sec. \ref{Sec-Friction}, we show how
internal variables emerge in a system's Hamiltonian. We particularly pay
attention to the relative velocity that induces a force that is dissipative,
\textit{i.e., }frictional and contributes to the dissipative work. In Sec.
\ref{Sec-Theory}, we finally develop our theoretical framework for NEQ
statistical thermodynamics and introduce statistical formulation of various
quantities. In Sec. \ref{Sec-WorkTheorem}, we propose a new work theorem both
for an exclusive and an inclusive Hamiltonian, which differs from the
Jarzynski formulation. The latter is well known for its failure for free
expansion. Therefore, we apply the new theorem to free expansion and show that
it holds in Sec. \ref{Sec-FreeExpansion}. Indeed, we argue that the new
theorem works for any arbitrary process and not just the free expansion. In
the last section, we give a brief summary of our results.

\section{Generalization of Prigogine's Partition\label{Sec-Notation}}

Any extensive SI quantity $X(t)$ of $\Sigma$, see Fig. \ref{Fig_System}, can
undergo two distinct kinds of changes in time: one due to the exchange with
the medium and another one due to internal processes. Following modern
notation \cite{Prigogine,deGroot}, exchanges of $X$ (we will suppress the time
argument unless clarity is needed) with the medium and changes within the
system carry the suffix e and i, respectively:%
\begin{equation}
dX\equiv d_{\text{e}}X+d_{\text{i}}X; \label{X-partition}%
\end{equation}
here
\[
dX(t)\doteq X(t+dt)-X(t).
\]
For the medium, we must replace $X$ by $\widetilde{X}$ so that $d\widetilde
{X}(t)=\widetilde{X}(t+dt)-\widetilde{X}(t)$. For $\Sigma_{0}$, we must use
$X_{0}$ so that $dX_{0}(t)\doteq X_{0}(t+dt)-X_{0}(t)$. We will assume
additivity of $X$ for $\Sigma_{0}$:\
\[
X_{0}=X(t)+\widetilde{X}(t).
\]
For this to hold, we need to assume that $\Sigma$ and $\widetilde{\Sigma}$
interact so \emph{weakly} that their interactions can be neglected; recall
that $E$ is one of the possible $X$. As there is no irreversibility within
$\widetilde{\Sigma}$ , we must have $d_{\text{i}}\widetilde{X}(t)=0$ for any
medium quantity $\widetilde{X}(t)$ and
\begin{equation}
d_{\text{e}}X\doteq-d\widetilde{X}=-d_{\text{e}}\widetilde{X}.
\label{exchange-equality}%
\end{equation}
It follows from additivity for $\Sigma_{0}$\ that
\begin{equation}
dX_{0}\equiv dX+d\widetilde{X}=d_{\text{i}}X.
\label{Isolated-system-irreversibility}%
\end{equation}
This means that any irreversibility in $\Sigma_{0}$ is ascribed to $\Sigma$,
and not to $\widetilde{\Sigma}$, as noted above.\ In a reversible change,
$d_{\text{i}}X\equiv0$. The additivity of entropy%
\[
S_{0}(t)=S(t)+\widetilde{S}(t)
\]
requires $\Sigma$ and $\widetilde{\Sigma}$ \ to be quasi-independent as said
earlier so that
\[
dS_{0}=dS+d\widetilde{S}.
\]

As an example of the Prigogine's notation, the entropy change $dS$ is written
as
\[
dS\equiv d_{\text{e}}S+d_{\text{i}}S
\]
for $\Sigma$; here,
\[
d_{\text{e}}S=-d_{\text{e}}\widetilde{S}%
\]
is the entropy exchange with the medium and
\begin{equation}
d_{\text{i}}S\geq0 \label{SecondLaw}%
\end{equation}
is \emph{irreversible entropy generation} due to internal processes within
$\Sigma$ and is nonnegative in accordance with the second law. It follows from
Eq. (\ref{Isolated-system-irreversibility}) that $d_{\text{i}}S$ is also the
entropy change $dS_{0}$\ of $\Sigma_{0}$. For the energy change $dE$, we
write
\begin{equation}
dE\equiv d_{\text{e}}E+d_{\text{i}}E, \label{E-Partition}%
\end{equation}
compare with Eq. (\ref{E_k-partition}), except that
\begin{equation}
d_{\text{i}}E\equiv0, \label{diE-EQ}%
\end{equation}
since no internal process, even chemical reaction, can change the energy of
the system. [The surprising fact is that $d_{\text{i}}E_{k}\neq0$, as we will
establish below; see Eq. (\ref{Works-Microstates-Internal}).] Similarly, if
$dW$ and $dQ$ represent the SI work done by and the heat change of the system,
then
\[
dW\equiv d_{\text{e}}W+d_{\text{i}}W,dQ\equiv d_{\text{e}}Q+d_{\text{i}}Q.
\]
Here, $d_{\text{e}}W$ and $d_{\text{e}}Q$ are the \emph{work exchange} and
\emph{heat exchange }with the medium, respectively, and $d_{\text{i}}W\equiv
d_{\text{i}}W_{0}\ $and $d_{\text{i}}Q\equiv$ $d_{\text{i}}Q_{0}$ are
\emph{irreversible} work done and heat generation due to internal processes in
$\Sigma$. For an isolated system such as $\Sigma_{0}$, the exchange quantity
vanishes so that $dX_{0}=d_{\text{i}}X_{0}$. In particular,
\begin{equation}
dS_{0}=d_{\text{i}}S_{0};dE_{0}=d_{\text{i}}E_{0}\equiv0,dW_{0}=d_{\text{i}%
}W_{0};dQ_{0}=d_{\text{i}}Q_{0}. \label{Isolated-system-irreversibility-W-Q}%
\end{equation}
Thus, a description in terms of SI quantity alone can describe an isolated and
an open system, whether in equilibrium or not. This makes the SI description
highly desirable.

The use of exchange quantities results in an MI-description for $\Sigma$. No
internal variables are required to determine the exchange quantities since
their affinities for $\widetilde{\Sigma}$ vanish \cite{deGroot,Prigogine},
which is always assumed to be in equilibrium. The SI-description always refers
to SI quantities ($dW$,$dQ$, etc.). The SI quantities may include internal
variables as their affinities need not vanish for $\Sigma$; in many cases,
they may not be readily measurable or even identifiable and require care in
interpreting results \cite{Maugin}. Therefore, the use of exchange quantities
(MI) is quite widespread. Despite this, the SI-description, which we call the
\emph{internal approach}, is more appropriate to study nonequilibrium
processes at the microscopic level, even if we do not use $\xi$, since
$\mathcal{H}$, itself an SI quantity, plays a central role. In contrast, the
MI-description is called the \emph{external approach} here.

\textbf{Intuitive Meaning}: Intuitively, $dW$ denotes the work done by the
system, a part $d_{\text{e}}W$ ($=P_{0}dV$) of which is transferred to
$\widetilde{\Sigma}_{\text{w}}$ through exchange and $d_{\text{i}}W$ is spent
to overcome \emph{internal} dissipative forces. Of $dQ$, a part $d_{\text{e}%
}Q$ is transferred from $\widetilde{\Sigma}_{\text{h}}$ through exchange and
$d_{\text{i}}Q=d_{\text{i}}W$ is generated by \emph{internal} dissipative forces.

The modern notation requires three different operators $d,d_{\text{e}}$ and
$d_{\text{i}}$ having the \emph{additive property}%
\begin{equation}
d\equiv d_{\text{e}}+d_{\text{i}}, \label{Differentiation-Additive}%
\end{equation}
acting on any extensive quantity. They are linear differential operators,
which satisfy%
\begin{align*}
d_{\alpha}(aX+bX^{\prime})  &  =ad_{\alpha}X+bd_{\alpha}X^{\prime},\\
d_{\alpha}(XX^{\prime})  &  =Xd_{\alpha}X^{\prime}+(d_{\alpha}X)X^{\prime};
\end{align*}
here $d_{\alpha}$ stands for $d,d_{\text{e}}$ and $d_{\text{i}}$, $X$ and
$X^{\prime}$ are two extensive quantities, and $a$ and $b$ are two pure
numbers. These operators also work on microstate quantities as well and will
prove very useful. We have already seen its use for $E_{k}$ in Eq.
(\ref{E_k-partition}).

Treating microstate probabilities $\left\{  p_{k}\right\}  $ as fractions
$\mathcal{N}_{k}/\mathcal{N}$ of the number of replicas in the $k$th
microstate out of a large but fixed number $\mathcal{N}$ of replicas of the
same system, we can also apply $d_{\alpha}$ on $p_{k}$ with the result%
\begin{equation}
dp_{k}=d_{\text{e}}p_{k}+d_{\text{i}}p_{k}, \label{Probability-partition}%
\end{equation}
that will be very useful later.

\section{New\ Concepts\label{Sec-NewConcepts}}

\subsection{SI- and MI-works}

In a NEQ state, the thermodynamics of the medium and the system are very
different. As pointed above, it is $\Sigma$\ that experiences dissipation
\cite{Fermi,Kestin,Woods}, which is partly responsible for the irreversible
entropy generation $d_{\text{i}}S$ in the system as we will see in Eq.
(\ref{Irreversible EntropyGeneration-Complete}). To clarify the difference
between the SI-description and the MI-description, we consider the
pressure-volume work and the work done by some internal variable $\xi$ as an
example; see the gas confined in a cylinder with one end closed and a movable
piston on the other end in Fig. \ref{Fig_Piston-Spring}(a). The internal
variable may describe the nonuniformity of the gas in an irreversible process.
It should be recalled that a system in EQ must be uniform
\cite{Gujrati-II,Landau} so any nonuniformity will result in a NEQ state. The
work in the two descriptions is the
\begin{equation}
\text{SI-work }dW=PdV+Ad\xi\label{SI-Work}%
\end{equation}
done \emph{by} $\Sigma$ (SI)\ or the%
\begin{equation}
\text{MI-work}\emph{\ }d\widetilde{W}=-P_{0}dV \label{MI-Work}%
\end{equation}
done \emph{by} $\widetilde{\Sigma}$ (MI) in terms of the instantaneous
pressure $P$ of $\Sigma$ or $P_{0}~$of $\widetilde{\Sigma}$, and their volume
change $dV$ or $-dV$, respectively, where we have used the additivity of the
volume $V_{0}=V+\widetilde{V}$ and the fact that $d\widetilde{V}=-dV$ to
ensure a constant $V_{0}$. Notice that $d\widetilde{W}$ does not contain any
contribution from the internal variable since $A_{0}=0$
\cite{deGroot,Prigogine}. It represents the MI-work done on the system by
$\widetilde{\Sigma}$ and its negative denotes the exchange work between
$\widetilde{\Sigma}$ and $\Sigma$ as the MI-\emph{work done by }$\Sigma$\emph{
against }$\widetilde{\Sigma}$; we denote it by
\[
d_{\text{e}}W=-d\widetilde{W}=P_{0}dV.
\]%
%TCIMACRO{\FRAME{ftbpFU}{3.5492in}{1.8273in}{0pt}{\Qcb{We schematically show a
%system of (a) gas in a cylinder with a movable piston under an external
%pressure $P_{0\text{ }}$controlling the volume $V$ of the gas, and (b) a
%particle attached to a spring in a fluid being pulled by an external force
%$F_{0}$, which causes the spring to stretch or compress depending on its
%direction. In an irreversible process, the internal pressure $P$ (the spring
%force $F_{\text{s}}$) is different in magnitude from the external pressure
%$P_{0}$ (external force $F_{0}$).}}{\Qlb{Fig_Piston-Spring}}%
%{Piston-Spring-Mod.eps}{\special{ language "Scientific Word";
%type "GRAPHIC";  maintain-aspect-ratio TRUE;  display "USEDEF";
%valid_file "F";  width 3.5492in;  height 1.8273in;  depth 0pt;
%original-width 5.418in;  original-height 2.7691in;  cropleft "0";
%filename 'Piston-Spring-Mod.eps';file-properties "XNPEU";}} }%
%BeginExpansion
\begin{figure}
[ptb]
\begin{center}
\includegraphics[
height=1.8273in,
width=3.5492in
]%
{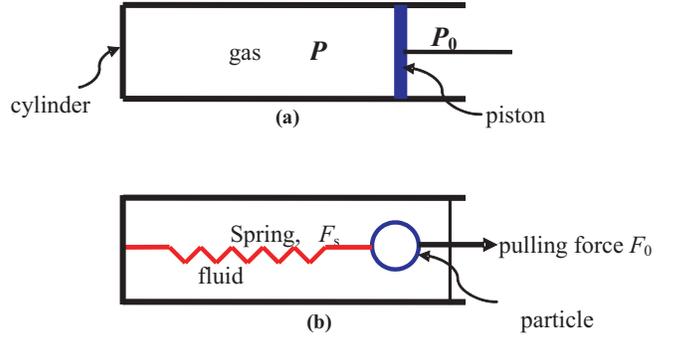}%
\caption{We schematically show a system of (a) gas in a cylinder with a
movable piston under an external pressure $P_{0\text{ }}$controlling the
volume $V$ of the gas, and (b) a particle attached to a spring in a fluid
being pulled by an external force $F_{0}$, which causes the spring to stretch
or compress depending on its direction. In an irreversible process, the
internal pressure $P$ (the spring force $F_{\text{s}}$) is different in
magnitude from the external pressure $P_{0}$ (external force $F_{0}$).}%
\label{Fig_Piston-Spring}%
\end{center}
\end{figure}
%EndExpansion

Let $d\widetilde{Q}$ denote the heat added to the medium $\widetilde{\Sigma}$,
here by the system. Then, $d_{\text{e}}Q=-d\widetilde{Q}$ is the heat
transferred (exchanged) to $\Sigma$ by $\widetilde{\Sigma}$ and $d_{\text{e}%
}W=-d\widetilde{W}$ is the work transferred (exchanged) by $\Sigma$ to
$\widetilde{\Sigma}$. Both $d_{\text{e}}Q$ and $d_{\text{e}}W$ refer to energy
flows across the boundary between $\Sigma$ and $\widetilde{\Sigma}$ (see the
exchange $d_{\text{e}}X$ in Fig. \ref{Fig_System}), and are determined by the
process of transfer that depends on the combined properties of $\Sigma$ and
$\widetilde{\Sigma}$. They are customarily used to express the first law for
$\Sigma$
\begin{equation}
dE=d_{\text{e}}Q-d_{\text{e}}W \label{FirstLaw-MI}%
\end{equation}
Our sign convention is that $d_{\text{e}}Q$ is positive when it is added to
$\Sigma$, and $d_{\text{e}}W$ is positive when it is transferred to
$\widetilde{\Sigma}$. Once, work has been identified, the use of the first law
allows us to uniquely determine the heat.

\subsection{Microwork}

The above macroscopic concepts of work represent thermodynamic averages (this
will become more clear later) \cite{Landau,Kestin,Woods,Gibbs}. This allows us
to identify the microscopic analogs $dW_{k}$ and $d\widetilde{W}_{k}$\ of $dW$
and $d\widetilde{W}$, respectively, at the level of $k$th microstates
$\mathfrak{m}_{k}$. The thermodynamic average connection, see also Sec.
\ref{Sec-MicrostateThermodynamics}, is the following:%
\begin{equation}
dW=\left\langle dW\right\rangle \doteq%
%TCIMACRO{\tsum \nolimits_{k}}%
%BeginExpansion
{\textstyle\sum\nolimits_{k}}
%EndExpansion
p_{k}dW_{k},d\widetilde{W}=\left\langle d\widetilde{W}\right\rangle \doteq%
%TCIMACRO{\tsum \nolimits_{k}}%
%BeginExpansion
{\textstyle\sum\nolimits_{k}}
%EndExpansion
p_{k}d\widetilde{W}_{k}, \label{AverageWorks}%
\end{equation}
which also introduces thermodynamic averaging $\left\langle \cdot\right\rangle
$ with respect to the instantaneous microstate probabilities $\left\{
p_{k}\right\}  $. For the energy change $dE_{k}$, we have its average given by%
\begin{equation}
dE_{S}=\left\langle dE\right\rangle \doteq%
%TCIMACRO{\tsum \nolimits_{k}}%
%BeginExpansion
{\textstyle\sum\nolimits_{k}}
%EndExpansion
p_{k}dE_{k}; \label{dE_isentropic}%
\end{equation}
the reason for the special symbol becomes once we recall the definition of
entropy $S$:%
\begin{equation}
S=-\left\langle \ln p\right\rangle \doteq-%
%TCIMACRO{\tsum \nolimits_{k}}%
%BeginExpansion
{\textstyle\sum\nolimits_{k}}
%EndExpansion
p_{k}\ln p_{k}. \label{Entropy}%
\end{equation}
As the above three averages are at fixed $p_{k}$'s, they are evaluated at
fixed entropy. This is shown by the suffix $S$ in $dE_{S}$. This quantity will
turn out to have an important role in our theory.

We will refer to $dW_{k}$ in the following as \emph{microstate work} or
\emph{microwork} in short (but not $d\widetilde{W}_{k}$). One can then
determine the accumulated work $\Delta W_{k},\Delta\widetilde{W}_{k}$\ and
$\Delta E_{k}$ along some path $\gamma$ followed by the microstate in some
process $\mathcal{P}_{0}$ by simply integrating $dW_{k}$ but without the
probabilities:%
\begin{equation}
\Delta W_{k}\doteq%
%TCIMACRO{\tint \nolimits_{\gamma}}%
%BeginExpansion
{\textstyle\int\nolimits_{\gamma}}
%EndExpansion
dW_{k},\Delta\widetilde{W}_{k}\doteq%
%TCIMACRO{\tint \nolimits_{\gamma}}%
%BeginExpansion
{\textstyle\int\nolimits_{\gamma}}
%EndExpansion
d\widetilde{W}_{k},\Delta E_{k}\doteq%
%TCIMACRO{\tint \nolimits_{\gamma}}%
%BeginExpansion
{\textstyle\int\nolimits_{\gamma}}
%EndExpansion
dE_{k}. \label{CumulativeWorks}%
\end{equation}
We observe that $dW_{k},d\widetilde{W}_{k}$ and $dE_{k}$\ are defined
regardless of the probability $p_{k}$. Because of this, we wish to emphasize
the following important point, which we present as a

\begin{condition}
\label{Cond-MicroWork_prob}The cumulative quantities $\Delta W_{k}%
,\Delta\widetilde{W}_{k}$ or $\Delta E_{k}$ exist or are defined regardless of
$p_{k}$. Thus, they exist or defined even if $p_{k}=0$.
\end{condition}

This condition will prove useful when we consider free expansion of a
classical gas in Sec. \ref{Sec-FreeExpansion}.

Traditional formulation of NEQ statistical mechanics and thermodynamics
\cite{Landau,Gibbs} takes a mechanical approach in which $\Sigma$ follows its
classical or quantum mechanical evolution dictated by its SI-Hamiltonian
$\mathcal{H}(V,\xi)$ \cite{Note-xi-dependence}; the interaction with
$\widetilde{\Sigma}$ is usually treated as a very weak stochastic
perturbation. However, investigating microwork (SI approach) is very
convenient as $p_{k}$'s do not change. Unfortunately, this description has
been overlooked by the current practitioners in the field who have
consistently used a MI-description for work by using the external approach but
then mistakenly confuse it with its SI-analog; see Eq. (\ref{MistakenIdentity}).

\subsection{Jarzynski Process\label{Sec-JarzynskiProcess}}

As $\Delta\widetilde{W}_{k}$ is not affected by the internal state ($P,A$, and
$\left\{  p_{k}\right\}  $) of $\Sigma$, it is a choice candidate for studying
work fluctuations and for the Jarzynski process. By setting the MI-works
$d\widetilde{W}_{k},\Delta\widetilde{W}_{k}$ equal to the SI-changes
$dE_{k},\Delta E_{k}$, respectively,
\begin{equation}
d\widetilde{W}_{k}=dE_{k},\Delta\widetilde{W}_{k}=\Delta E_{k},
\label{JarzynskiWorkEnergy}%
\end{equation}
[compare with Eq. (\ref{MistakenIdentity})] Jarzynski developed the following
famous nonequilibrium work relation, the JE, using a very clever averaging
$\left\langle \cdot\right\rangle _{0}$
\begin{equation}
\left\langle e^{-\beta_{0}\Delta\widetilde{W^{\prime}}}\right\rangle
_{0}\doteq%
%TCIMACRO{\tsum \nolimits_{k}}%
%BeginExpansion
{\textstyle\sum\nolimits_{k}}
%EndExpansion
p_{k0}e^{-\beta_{0}\Delta\widetilde{W^{\prime}}_{k}}=e^{-\beta_{0}\Delta
F^{\prime}}~ \label{JarzynskiRelation}%
\end{equation}
for a Jarzynski process (we use a prime now reserved for the inclusive
approach discussed in Sec. \ref{Sec-Exclusive-Inclusive-Hamiltonians}). In Eq.
(\ref{JarzynskiRelation}), where an externally driven nonequilibrium process
$\mathcal{P}_{0}$ between two equilibrium states, the initial state
\textsf{A}$_{\text{eq}}$ and the final state \textsf{B}$_{\text{eq}}$, at the
same inverse temperature $\beta_{0}$ is considered, $\left\langle
\cdot\right\rangle _{0}$ refers to a special averaging with respect to the
canonical probabilities of the initial state \textsf{A}$_{\text{eq}}$:
\[
p_{k0}^{\prime}\doteq e^{-\beta_{0}E_{k}^{\prime}}/Z_{\text{in}}^{\prime
}(\beta_{0});
\]
here, the suffix $0$ refers to the initial state \textsf{A}$_{\text{eq}}$,
$Z_{\text{in}}^{\prime}(\beta_{0})$ the initial equilibrium partition function
for the system, $\Delta\widetilde{W^{\prime}}_{k}$ the cumulative work done on
$\mathfrak{m}_{k}$ during $\mathcal{P}_{0}$, and $\Delta F^{\prime}$ the
change in the free energy over $\mathcal{P}_{0}$
\cite{Gujrati-jarzynski-SecondLaw}. However, the average $\left\langle
\cdot\right\rangle _{0}$ in Eq. (\ref{JarzynskiRelation}) does not represent a
thermodynamic average [compare with Eq. (\ref{Av-ExternalWork-Medium})] over
$\mathcal{P}_{0}$ as recently pointed out \cite{Gujrati-jarzynski-SecondLaw}.
The inverse temperature along $\mathcal{P}_{0}$ may not always exist or may be
different than $\beta_{0}$ due to irreversibility \cite{Cohen,Muschik}. The
process $\mathcal{P}_{0}$ consists of two stages: in the \emph{driving stage}
$\mathcal{P}$, external work is done ($\Delta\widetilde{W}_{k}^{\prime}\neq
0$), and the \emph{reequilibration stage} $\overline{\mathcal{P}}%
\doteq\mathcal{P}_{0}\backslash\mathcal{P}$ ($\mathcal{P}_{0}=\mathcal{P\cup
}\overline{\mathcal{P}}$) in which $\Sigma$ is in thermal contact with the
heat bath $\widetilde{\Sigma}_{\text{h}}$ at $\beta_{0}$, during which no work
is done ($\Delta\widetilde{W}_{k}^{\prime}\equiv0$); $\widetilde{\Sigma
}_{\text{h}}$ may or may not be present during $\mathcal{P}$. The equality is
not obeyed for free expansion of a classical gas for which $\Delta
\widetilde{W^{\prime}}_{k}\equiv0$ but $\Delta F^{\prime}\neq0$ as previously
observed \cite{Sung,Gross,Jarzynski-Gross}.

\section{The First and Second Laws, Thermodynamic Forces and
Dissipation\label{Sec-FirstLaw-ThermodynamicForces}}

\subsection{Gibbs Fundamental Relation\label{Sec-GibbsRelation}}

For the IEQ entropy $S(E,V,\xi)$, the Gibbs fundamental relation \cite{Gibbs}
is given by%
\begin{equation}
dS=(dE+PdV+Ad\xi)/T, \label{Gibbs-S-IEQ}%
\end{equation}
where
\[
1/T=\partial S/\partial E,P=T\partial S/\partial V,A=T\partial S/\partial\xi
\]
are the inverse temperature, pressure and the affinity of the system that
represent its fields. As $S$ is a SI-quantity, the above fields are also
SI-quantities. In EQ, the affinity vanishes, which means that $S$ no longer
depends on $\xi$. In other words, $\xi$ is no longer independent of $E$ and
$V$ in EQ so we can write it as $\xi_{\text{eq}}\equiv\xi(E,V)$. Otherwise,
$\xi$ is an independent variable in NEQ states. We can rewrite Eq.
(\ref{Gibbs-S-IEQ}) as%
\begin{equation}
dE=TdS-PdV-Ad\xi. \label{Gibbs-E-IEQ}%
\end{equation}
We see that the first term on the right side of the equation represents the
change in the energy due to the entropy, while the last two terms represent
the \emph{isentropic} change $dE_{S}$ in the energy; see Eq.
(\ref{dE_isentropic}). These changes represent SI-changes in $E$ and play an
important role in our approach. In particular, as $V$ and $\xi$ are work
variables that we collectively denote by $\mathbf{W}$, we see from Eq.
(\ref{SI-Work}) that $\left.  dE\right\vert _{S}$ is nothing but the SI-work
$dW$. Thus,%
\begin{equation}
TdS=dE-dE_{S}\equiv dE_{\mathbf{W}}\text{,} \label{IsometricChange}%
\end{equation}
where $dE_{\mathbf{W}}$ represents the change $dE$ at fixed $\mathbf{W}$. For
brevity, we will call $dE_{\mathbf{W}}$ the \emph{isometric} change in which
all work variables are held fixed. Thus, we arrive at the following conclusion:

\begin{conclusion}
\label{Conc-Isentropic-Isometric-Changes} The change $dE$ consists of two
independent contributions- an isentropic change $dW=dE_{S}$, and an isometric
change $TdS=dE_{\mathbf{W}}$.
\end{conclusion}

The conclusion allows us to express the first law in an alternative form using
SI-quantities as we discuss below.

\subsection{The First Law\label{Sec-FirstLaw}}

Using the exchange work $d_{\text{e}}W=-d\widetilde{W}=P_{0}dV$, the (MI)
\emph{work done by }$\Sigma$\emph{ against }$\widetilde{\Sigma}$, and the work
partition%
\begin{equation}
dW=d_{\text{e}}W+d_{\text{i}}W, \label{GeneralizedWork}%
\end{equation}
we identify $d_{\text{i}}W$. We also introduce the exchange heat $d_{\text{e}%
}Q=-d\widetilde{Q}=T_{0}d_{\text{e}}S$ between $\widetilde{\Sigma}$ and
$\Sigma$, which represents a MI-quantity. Defining
\begin{equation}
d_{\text{i}}Q\doteq d_{\text{i}}W \label{diQ-diW-EQ}%
\end{equation}
to emphasize the well-known fact that internal work can also be treated as
internal heat, we introduce%
\begin{equation}
dQ=d_{\text{e}}Q+d_{\text{i}}Q \label{GeneralizedHeat}%
\end{equation}
so that the conventional MI-form of the first law in Eq. (\ref{FirstLaw-MI})
can also be written in an SI-form
\begin{equation}
dE=dQ-dW. \label{FirstLaw-SI}%
\end{equation}
From now on, we will refer to $dW$ and $dQ$ ($d_{\text{e}}W$ and $d_{\text{e}%
}Q$) as \emph{generalized work} and \emph{generalized heat} (exchange work and
exchange heat), respectively, so that no confusion can arise.

If we compare this form with Eq. (\ref{Gibbs-E-IEQ}), we recognize that the
last two terms there is $dW$ above. Hence, the first term $TdS$ there is
nothing but $dQ$ introduced above. This provides us with an important
consequence of our approach
\cite{Gujrati-Heat-Work0,Gujrati-Heat-Work,Gujrati-I,Gujrati-II,Gujrati-III,Gujrati-Entropy1,Gujrati-Entropy2}
in which
\begin{equation}
dQ=TdS; \label{ClausiusEquality}%
\end{equation}
compare with Eq. (\ref{IsometricChange}). In other words, $dQ$ is nothing but
the isometric change $\left.  dE\right\vert _{\mathbf{W}}$ in $E$: it is the
change when no work is done ($dW=0$). The relation $dQ=TdS$ is very
interesting in that it not only turns the well-known Clausius inequality
$d_{\text{e}}Q=T_{0}d_{\text{e}}S\leq$ $T_{0}dS$ into an equality but also
makes the generalized work $dW$ in Eq. (\ref{FirstLaw-SI}) isentropic, whereas
$d_{\text{e}}W=-d\widetilde{W}$ is not. This aspect of $dW$ will prove useful
later. Moreover, Eq. (\ref{FirstLaw-SI}) allows us to uniquely identify
generalized heat and work; this work is always isentropic and this heat is
always isometric. On the other hand, the MI-heat and the MI- work suffer from
ambiguity; see, for example, Kestin \cite{Kestin}. We summaries these
important observations here as

\begin{conclusion}
\label{Conclusion-GeneralizedWorkHeat} The generalized work $dW$ is isentropic
change in the energy, while the generalized heat $dQ=TdS$ is the isometric
change in the energy.
\end{conclusion}

It is important to draw attention to the following important fact. We first
recognize that the first law in Eq. (\ref{FirstLaw-MI}) really refers to the
change in the energy caused by exchange quantities. Therefore, $dE$ on the
left truly represents $d_{\text{e}}E$; see Eq. (\ref{E-Partition}).
Accordingly, we write Eq. (\ref{FirstLaw-MI}) as%
\begin{equation}
d_{\text{e}}E=d_{\text{e}}Q-d_{\text{e}}W. \label{First-Law-deE}%
\end{equation}
Subtracting this equation from Eq. (\ref{FirstLaw-SI}), we obtain the identity%
\begin{equation}
d_{\text{i}}E=d_{\text{i}}Q-d_{\text{i}}W\equiv0, \label{First-Law-diE}%
\end{equation}
where we have used Eq. (\ref{diE-EQ}). The above equation is nothing but the
identity in Eq. (\ref{diQ-diW-EQ}). However, the analysis also demonstrates
the important fact that the first law can be applied either to the exchange
process or to the interior process. In the last formulation, it is also
applicable to an isolated system.

The analog of $d_{\text{i}}W$ between two nearby \emph{equilibrium} states is
well known in classical thermodynamics and is usually called the
\emph{dissipated} work; see for example, p. 12 in Woods \cite{Woods}. The
dissipated work $(P-P_{0})dV$ is the difference between the generalized work
$dW=PdV$ and the reversible work $dW_{\text{rev}}=P_{0}dV$ done by $\Sigma
$\ between the same two equilibrium states and is always nonnegative in
accordance with Theorem \ref{Theorem-Work-Energy-Principle} that will be
proved later in its general form (which refers to $d_{\text{i}}W$ between any
two nearby \emph{arbitrary} states that need not necessarily be equilibrium
states). If they are not, then the procedure described by Wood \ to determine
dissipated work cannot work as there is no reversible path connecting two
nonequilibrium states. This makes our approach for introducing $d_{\text{i}%
}W=dW-d_{\text{e}}W$ much more general.

\subsection{The Second Law}

We rewrite $dE$ in Eq. (\ref{FirstLaw-SI}) as
\[
dW=TdS-dE=-dF+(T-T_{0})dS,
\]
where
\begin{equation}
F=E-T_{0}S \label{HelmholtzFreeEnergy}%
\end{equation}
is the Helmholtz free energy; note the presence of $T_{0}$ and not $T$ in $F$.
As $(T-T_{0})dS\leq0$ established later, see Eq. (\ref{SecondLaw-Consequences}%
), we conclude
\begin{equation}
dW\leq-dF, \label{SecondLaw-Work-SI}%
\end{equation}
as a consequence of the second law. It is also easy to see from this that
\begin{equation}
d\widetilde{W}\geq dF+d_{\text{i}}W\geq dF, \label{SecondLaw-Work-MI}%
\end{equation}
where we have used the fact that $d_{\text{i}}W\geq0$, which is again a
consequence of the second law as shown in Theorem
\ref{Theorem-Work-Energy-Principle}. Similar inequalities also hold in the
inclusive approach:
\begin{equation}
dW^{\prime}\leq-dF^{\prime},d\widetilde{W^{\prime}}\geq dF^{\prime}.
\label{SecondLaw-Work-Prime}%
\end{equation}

\subsection{Dissipation and Thermodynamic Forces\label{Sec-Dissipation}}

We first turn to the relationship of dissipation with the entropy generation
$d_{\text{i}}S$ for $\Sigma$. The general result for the present case is%
\begin{equation}
Td_{\text{i}}S=\frac{(T_{0}-T)}{T_{0}}d_{\text{e}}Q+(P-P_{0})dV+Ad\xi\geq0;
\label{Irreversible EntropyGeneration-Complete}%
\end{equation}
see, for example, Ref. \cite{Prigogine}. In accordance with the second law in
Eq. (\ref{SecondLaw}), each term on the right of the first equation must be
nonnegative. The first term is due to heat exchange $d_{\text{e}}Q$ with
$\widetilde{\Sigma}$ at different temperatures, the second term is the
irreversible work due to pressure imbalance and the third term is the
irreversible work due to $\xi$. It is customary to think of $A$ as the
instantaneous "force" associated with the "displacement" $d\xi$. As the
affinity $A_{0}=0$ for $\widetilde{\Sigma}$, $A-A_{0}\equiv A$ also denotes
the affinity imbalance. Therefore, the last two terms above denote dissipated
work, which is customarily called dissipation; the first term is not
considered part of it. Therefore, in this work, we will use dissipation to
denote the dissipated work.

It is clear that the root cause of dissipation is a "\emph{force imbalance}"
$P(t)-P_{0},A(t)-A_{0}\equiv A(t)$, etc.
\cite{Kestin,Woods,Gujrati-Heat-Work0,Gujrati-Heat-Work,Gujrati-I,Gujrati-II,Gujrati-III,Gujrati-Entropy1,Gujrati-Entropy2,Gujrati-GeneralizedWork}
between the external and the internal forces performing work, giving rise to
an internal work $d_{\text{i}}W$ due to all kinds of force imbalances, which
is not captured by using the MI-work $d\widetilde{W}$ unless we recognize that
there must be some \emph{nonvanishing} force imbalance to cause dissipation.
The irreversible or dissipated work is \cite{Kestin,Woods,Prigogine}
\begin{equation}
d_{\text{i}}W=(P-P_{0})dV+Ad\xi\geq0, \label{Work-Irreversible}%
\end{equation}
which is generated within $\Sigma$; we give a general proof of this result
(see Eq. (\ref{SecondLaw-IrriversibleWork})) in Theorem
\ref{Theorem-Work-Energy-Principle}. The pressure and affinity imbalance
$P-P_{0}\ $and $A-A_{0}\equiv A$ are commonly known as \emph{thermodynamic
forces} driving the system towards equilibrium.

If we include the relative velocity between the subsystem $\Sigma_{\text{p}}$
formed by the piston and the subsystem $\Sigma_{\text{gc}}$ of the gas and the
cylinder ($\Sigma=\Sigma_{\text{p}}\cup\Sigma_{\text{gc}}$), we must account
for \cite{Gujrati-II} an additional term $-\mathbf{V\cdot}d\mathbf{P}%
_{\text{p}}$ in $d_{\text{i}}W$ due to the relative velocity $\mathbf{V}$:
\begin{equation}
d_{\text{i}}W=(P-P_{0})dV-\mathbf{V\cdot}d\mathbf{P}_{\text{p}}+Ad\xi.
\label{Irreversible_Work-Piston-General}%
\end{equation}
This is reviewed in Sec. \ref{Sec-Friction}. We will come back to this term
later when we consider the motion of a particle attached to a spring; see Fig.
\ref{Fig_Piston-Spring}(b).

The irreversible work is present even if there is no temperature difference
such as in an isothermal process as long as there exists some nonzero
thermodynamic force. The resulting irreversible entropy generation is then
given by%
\begin{equation}
Td_{\text{i}}S=d_{\text{i}}W\geq0. \label{Irreversible EntropyGeneration}%
\end{equation}
We summarize this as a conclusion \cite{Prigogine}:

\begin{conclusion}
\label{Conclusion-ThermodynamicForce-Irreversibility}To have dissipation, it
is necessary and sufficient to have a nonzero thermodynamic force. In its
absence, there can be no dissipation regardless of the time-dependence of the
work process.
\end{conclusion}

\section{Microstate Thermodynamics\label{Sec-MicrostateThermodynamics}}

We have applied the internal approach microscopically to the set $\left\{
\mathfrak{m}_{k}\right\}  $ of microstates
\cite{Gujrati-Heat-Work0,Gujrati-Entropy1,Gujrati-Entropy2} to obtain a
\emph{microscopic representation} of (generalized) work and heat in terms of
the set of microstate probabilities $\left\{  p_{k}\right\}  $ and other
SI-quantities. We expand on this approach here and exploit it. As we will be
dealing with microstates, we will mostly use their energy set $\left\{
E_{k}\right\}  $ instead of $\mathcal{H}$.

\subsection{General Discussion}

As $\mathcal{H}$ depends on $V$ and $\xi$ as parameters, $E_{k}$ also depends
on them as parameters so we write it as $E_{k}(V,\xi)$. Then the thermodynamic
energy is given by
\begin{equation}
E=\left\langle E\right\rangle \doteq%
%TCIMACRO{\tsum \nolimits_{k}}%
%BeginExpansion
{\textstyle\sum\nolimits_{k}}
%EndExpansion
p_{k}E_{k}(V,\xi). \label{ThermodynamicEnergy}%
\end{equation}
The dependence of $p_{k}$ is not important. Now, we have%
\begin{equation}
dE=d\left\langle E\right\rangle =%
%TCIMACRO{\tsum \nolimits_{k}}%
%BeginExpansion
{\textstyle\sum\nolimits_{k}}
%EndExpansion
p_{k}dE_{k}+%
%TCIMACRO{\tsum \nolimits_{k}}%
%BeginExpansion
{\textstyle\sum\nolimits_{k}}
%EndExpansion
E_{k}dp_{k}; \label{Energy Differential}%
\end{equation}
here,
\begin{equation}
dE_{k}\doteq(\partial E_{k}/\partial V)dV+(\partial E_{k}/\partial\xi)d\xi.
\label{MicrostateEnergyChange}%
\end{equation}
As $\left\{  p_{k}\right\}  $ is not changed in the first sum, it is evaluated
at \emph{fixed entropy }of $\Sigma$ \cite{Gujrati-Heat-Work0,Gujrati-Entropy2}
so this isentropic sum must be identified with $-dW$:%
\begin{equation}
dW=-\left\langle dE\right\rangle \doteq-%
%TCIMACRO{\tsum \nolimits_{k}}%
%BeginExpansion
{\textstyle\sum\nolimits_{k}}
%EndExpansion
p_{k}dE_{k}; \label{GeneralizedWork0}%
\end{equation}
the summand must denote $dW_{k}$; compare with Eq. (\ref{AverageWorks}):%
\begin{equation}
dW_{k}=-dE_{k}. \label{dW-dE}%
\end{equation}
It follows from Eq. (\ref{E_k-partition}) that%
\begin{equation}
d_{\text{e}}W_{k}=-d_{\text{e}}E_{k},d_{\text{i}}W_{k}=-d_{\text{i}}E_{k}.
\label{deW-deE-diW-diE}%
\end{equation}
We also find that $dW$ only changes $E_{k}$'s but not $p_{k}$'s. The second
sum in $dE$ must be identified with the generalized heat $dQ$%
\begin{equation}
dQ\doteq%
%TCIMACRO{\tsum \nolimits_{k}}%
%BeginExpansion
{\textstyle\sum\nolimits_{k}}
%EndExpansion
E_{k}dp_{k}=TdS, \label{GeneralizedHeat0}%
\end{equation}
see Eq. (\ref{FirstLaw-SI}). It is possible to express $dQ$ also as an average
by introducing Gibbs \emph{index of probability} \cite{Gibbs}
\begin{equation}
\eta_{k}\doteq\ln p_{k}. \label{IndexProbability}%
\end{equation}
We have%
\begin{equation}
dQ=\left\langle dQ\right\rangle =\left\langle Ed\eta\right\rangle \doteq%
%TCIMACRO{\tsum \nolimits_{k}}%
%BeginExpansion
{\textstyle\sum\nolimits_{k}}
%EndExpansion
p_{k}E_{k}d\eta_{k}, \label{GeneralizedHeat-Av}%
\end{equation}
so that we can formally introduce a quantity
\begin{equation}
dQ_{k}\doteq E_{k}d\eta_{k} \label{MicrostateHeat}%
\end{equation}
for $\mathfrak{m}_{k}$, which we will refer to as \emph{microstate heat} or
\emph{microheat} in short. It also changes the probability at fixed $E_{k}$.
Thus, both $dQ$ and $dQ_{k}$ are stochastic quantities. We should emphasize
that our concept of microstate heat is very different from the concept of heat
currently used in the literature.

We also observe that the generalized heat $dQ$ and $dQ_{k}$ only changes
$p_{k}$'s, but not $E_{k}$'s. Therefore, the following aspects of the
generalized quantities will be central in our discussion later, which we
present as two conclusions:

\begin{conclusion}
\label{Conclusion-MicroWorkHeat}The microwork $dW_{k}$ changes $E_{k}$ without
changing $p_{k}$. Thus, a purely mechanical approach can be used for
microwork. The effect of microheat is to change $p_{k}$ but not $E_{k}$ so it
is microheat that makes a thermodynamic process stochastic by changing $p_{k}$.
\end{conclusion}

\begin{conclusion}
\label{Conclusion-Microheat-dE}While the microheat $dQ_{k}$ doe not change
$E_{k}$, it does contribute to the energy change $dE$ through $dQ=%
%TCIMACRO{\tsum \nolimits_{k}}%
%BeginExpansion
{\textstyle\sum\nolimits_{k}}
%EndExpansion
E_{k}dp_{k}$ as $p_{k}$'s change.
\end{conclusion}

The following conclusion also follows from the above general discussion.

\begin{conclusion}
Comparing $d\left\langle E\right\rangle $ in Eq. (\ref{Energy Differential})
with $\left\langle dE\right\rangle =-dW$ given above, we conclude that
\begin{equation}
d\left\langle E\right\rangle \neq\left\langle dE\right\rangle .
\label{dAvE-AvdE}%
\end{equation}
The difference is $dQ$ above.
\end{conclusion}

\subsection{Microscopic Force Imbalance}

The force imbalance necessary for irreversibility has its root in a similar
imbalance at the microstate level \cite{Gujrati-GeneralizedWork}. We only have
to recall Eq. (\ref{AverageWorks}) and to recognize that
\begin{equation}
P=%
%TCIMACRO{\tsum \nolimits_{k}}%
%BeginExpansion
{\textstyle\sum\nolimits_{k}}
%EndExpansion
p_{k}P_{k},A=%
%TCIMACRO{\tsum \nolimits_{k}}%
%BeginExpansion
{\textstyle\sum\nolimits_{k}}
%EndExpansion
p_{k}A_{k}, \label{Average-P-A}%
\end{equation}
where $P_{k}\doteq-\partial E_{k}/\partial V,A_{k}\doteq-\partial
E_{k}/\partial\xi$ are the instantaneous pressure and affinity associated with
$\mathfrak{m}_{k}$ at that time $t$. (We can similarly introduce
$\mathbf{V}_{k}$ if we are interested in it.) It follows from Eq.
(\ref{MicrostateEnergyChange}) that $dE_{k}=-dW_{k}$. Using Eqs. (\ref{dW-dE})
and (\ref{deW-deE-diW-diE}), along with (\ref{SI-Work}), (\ref{MI-Work}),
(\ref{GeneralizedWork}) and (\ref{Work-Irreversible}), we conclude that
\begin{subequations}
\label{Works-Microstates}%
\begin{align}
dW_{k}  &  =-dE_{k}=P_{k}dV+A_{k}d\xi,\label{Works-Microstates-total}\\
d\widetilde{W}_{k}  &  =-d_{\text{e}}W_{k}=d_{\text{e}}E_{k}=-P_{0}%
dV,\label{Works-Microstates-Exchange}\\
d_{\text{i}}W_{k}  &  =-d_{\text{i}}E_{k}=(P_{k}-P_{0})dV+A_{k}d\xi.
\label{Works-Microstates-Internal}%
\end{align}
The important point to note is that the force imbalances $P_{k}-P_{0}$ and
$A_{k}$ determine the internal changes $d_{\text{i}}E_{k}$ or $d_{\text{i}%
}W_{k}$ for $\mathfrak{m}_{k}$. On the other hand, $dE_{k}$ is determined by
the SI-fields $P_{k}$ and $A_{k}$, while $d_{\text{e}}E_{k}$ is determined by
the MI-field $P_{0}$. From now on, we will reserve the use of force imbalance
to denote it only at the microstate level. At this level, it appears as a
mechanical force imbalance. In contrast, we will refer to macroscopic force
imbalance from now on as thermodynamic forces. We will in the following see
that there are reasons to make a clear distinction between the two; compare
Conclusion \ref{Conclusion-ThermodynamicForce-Irreversibility} with Conclusion
\ref{Conclusion-ForceImbalance-Dissipation}.

The following comment for various microscopic works in Eq.
(\ref{Works-Microstates}) should be obvious: the average of Eq.
(\ref{Works-Microstates-total}) gives Eq. (\ref{SI-Work}), the average of Eq.
(\ref{Works-Microstates-Exchange}) yields Eq. (\ref{MI-Work}), and the average
of Eq. (\ref{Works-Microstates-Internal}) gives Eq. (\ref{Work-Irreversible}).
This helps us extend Conclusion
\ref{Conclusion-ThermodynamicForce-Irreversibility} for thermodynamic forces
to a new conclusion valid only for mechanical force imbalances $P_{k}-P_{0}$
and $A_{k}$ to accommodate the possibility when they result in%
\end{subequations}
\begin{equation}%
%TCIMACRO{\tsum \nolimits_{k}}%
%BeginExpansion
{\textstyle\sum\nolimits_{k}}
%EndExpansion
p_{k}(P_{k}-P_{0})=0,%
%TCIMACRO{\tsum \nolimits_{k}}%
%BeginExpansion
{\textstyle\sum\nolimits_{k}}
%EndExpansion
p_{k}A_{k}=0, \label{EQ-condition}%
\end{equation}
even if $P_{k}-P_{0}$ and $A_{k}$ are nonzero. In this case, the system is in
EQ despite the fact that the force imbalance is present. Indeed, we will see
later that the presence of force imbalance is a ubiquitous phenomenon and must
be incorporated even if they result in zero thermodynamic forces. The new
conclusion is the following:

\begin{conclusion}
\label{Conclusion-ForceImbalance-Dissipation}To have dissipation, it is
necessary but not sufficient to have a nonzero mechanical force imbalance.
Even in its presence, there may be no dissipation if Eq. (\ref{EQ-condition})
is satisfied. Nevertheless, it must be incorporated for a consistent theory.
\end{conclusion}

\subsection{Accumulation of thermodynamic work along a Trajectory}

Consider a path $\gamma$ taken by the $k$th microstate $\mathfrak{m}_{k}$
during a NEQ process $\mathcal{P}_{0}$ and $dW_{k}$ and $d\widetilde{W}_{k}$
during $t$\ and $t+dt$ along $\gamma$. The net works $\Delta W_{k}$ and
$\Delta\widetilde{W}_{k}$ is the integral over $\gamma$ as shown in Eq.
(\ref{CumulativeWorks}). The thermodynamic works are averages $\left\langle
\cdot\right\rangle $ over $\left\{  p_{k}\right\}  $ (not to be confused with
$\left\langle \cdot\right\rangle _{0}$ in Eq. (\ref{JarzynskiRelation})):%
\begin{align}
dW  &  =\left\langle dW\right\rangle =%
%TCIMACRO{\tsum \limits_{k}}%
%BeginExpansion
{\textstyle\sum\limits_{k}}
%EndExpansion
p_{k}dW_{k},~d\widetilde{W}=\left\langle d\widetilde{W}\right\rangle =%
%TCIMACRO{\tsum \limits_{k}}%
%BeginExpansion
{\textstyle\sum\limits_{k}}
%EndExpansion
p_{k}d\widetilde{W}_{k},\nonumber\\
\Delta W  &  =%
%TCIMACRO{\tint \nolimits_{\gamma}}%
%BeginExpansion
{\textstyle\int\nolimits_{\gamma}}
%EndExpansion
dW=%
%TCIMACRO{\tint \nolimits_{\gamma}}%
%BeginExpansion
{\textstyle\int\nolimits_{\gamma}}
%EndExpansion
\left\langle dW\right\rangle =%
%TCIMACRO{\tint \nolimits_{\gamma}}%
%BeginExpansion
{\textstyle\int\nolimits_{\gamma}}
%EndExpansion%
%TCIMACRO{\tsum \limits_{k}}%
%BeginExpansion
{\textstyle\sum\limits_{k}}
%EndExpansion
p_{k}dW_{k}~\mathbf{,}\label{Av-ExternalWork-Medium}\\
\Delta\widetilde{W}  &  =%
%TCIMACRO{\tint \nolimits_{\gamma}}%
%BeginExpansion
{\textstyle\int\nolimits_{\gamma}}
%EndExpansion
d\widetilde{W}=%
%TCIMACRO{\tint \nolimits_{\gamma}}%
%BeginExpansion
{\textstyle\int\nolimits_{\gamma}}
%EndExpansion
\left\langle d\widetilde{W}\right\rangle =%
%TCIMACRO{\tint \nolimits_{\gamma}}%
%BeginExpansion
{\textstyle\int\nolimits_{\gamma}}
%EndExpansion%
%TCIMACRO{\tsum \limits_{k}}%
%BeginExpansion
{\textstyle\sum\limits_{k}}
%EndExpansion
p_{k}d\widetilde{W}_{k}~\mathbf{.}\nonumber
\end{align}

\subsection{Exchange and Irreversible Components}

Let us focus on the exclusive approach and express exchange quantities
microscopically. Using the partition of $dE_{k}$ in $dW_{k}$, we have%
\[
d_{\text{e}}W=-%
%TCIMACRO{\tsum \nolimits_{k}}%
%BeginExpansion
{\textstyle\sum\nolimits_{k}}
%EndExpansion
p_{k}d_{\text{e}}E_{k},\ d_{\text{i}}W=-%
%TCIMACRO{\tsum \nolimits_{k}}%
%BeginExpansion
{\textstyle\sum\nolimits_{k}}
%EndExpansion
p_{k}d_{\text{i}}E_{k}\geq0.
\]
Similarly, using $dp_{k}\equiv d_{\text{e}}p_{k}+d_{\text{i}}p_{k}$ from Eq.
(\ref{Probability-partition}), we have
\[
d_{\text{e}}Q=%
%TCIMACRO{\tsum \nolimits_{k}}%
%BeginExpansion
{\textstyle\sum\nolimits_{k}}
%EndExpansion
E_{k}d_{\text{e}}p_{k},\ d_{\text{i}}Q=%
%TCIMACRO{\tsum \nolimits_{k}}%
%BeginExpansion
{\textstyle\sum\nolimits_{k}}
%EndExpansion
E_{k}d_{\text{i}}p_{k}\geq0,
\]
and
\[
d_{\text{e}}S=-%
%TCIMACRO{\tsum \nolimits_{k}}%
%BeginExpansion
{\textstyle\sum\nolimits_{k}}
%EndExpansion
\ln p_{k}d_{\text{e}}p_{k},\ \ d_{\text{i}}S=-%
%TCIMACRO{\tsum \nolimits_{k}}%
%BeginExpansion
{\textstyle\sum\nolimits_{k}}
%EndExpansion
\ln p_{k}d_{\text{i}}p_{k}\geq0.
\]
We finally have%
\begin{align}
\Delta_{\text{e}}W  &  =-%
%TCIMACRO{\tsum \nolimits_{k}}%
%BeginExpansion
{\textstyle\sum\nolimits_{k}}
%EndExpansion%
%TCIMACRO{\tint \nolimits_{\gamma}}%
%BeginExpansion
{\textstyle\int\nolimits_{\gamma}}
%EndExpansion
p_{k}d_{\text{e}}E_{k},\Delta_{\text{i}}W=-%
%TCIMACRO{\tsum \nolimits_{k}}%
%BeginExpansion
{\textstyle\sum\nolimits_{k}}
%EndExpansion%
%TCIMACRO{\tint \nolimits_{\gamma}}%
%BeginExpansion
{\textstyle\int\nolimits_{\gamma}}
%EndExpansion
p_{k}d_{\text{i}}E_{k}\geq0,\nonumber\\
\Delta_{\text{e}}Q  &  =%
%TCIMACRO{\tsum \nolimits_{k}}%
%BeginExpansion
{\textstyle\sum\nolimits_{k}}
%EndExpansion%
%TCIMACRO{\tint \nolimits_{\gamma}}%
%BeginExpansion
{\textstyle\int\nolimits_{\gamma}}
%EndExpansion
E_{k}d_{\text{e}}p_{k},\ \Delta_{\text{i}}Q=%
%TCIMACRO{\tsum \nolimits_{k}}%
%BeginExpansion
{\textstyle\sum\nolimits_{k}}
%EndExpansion%
%TCIMACRO{\tint \nolimits_{\gamma}}%
%BeginExpansion
{\textstyle\int\nolimits_{\gamma}}
%EndExpansion
E_{k}d_{\text{i}}p_{k}\geq0,\label{CorrectedForms}\\
\Delta W  &  =-%
%TCIMACRO{\tsum \nolimits_{k}}%
%BeginExpansion
{\textstyle\sum\nolimits_{k}}
%EndExpansion%
%TCIMACRO{\tint \nolimits_{\gamma}}%
%BeginExpansion
{\textstyle\int\nolimits_{\gamma}}
%EndExpansion
p_{k}dE_{k},\Delta Q=%
%TCIMACRO{\tsum \nolimits_{k}}%
%BeginExpansion
{\textstyle\sum\nolimits_{k}}
%EndExpansion%
%TCIMACRO{\tint \nolimits_{\gamma}}%
%BeginExpansion
{\textstyle\int\nolimits_{\gamma}}
%EndExpansion
E_{k}dp_{k},\nonumber
\end{align}
along with $\Delta_{\text{i}}W=\Delta_{\text{i}}Q$. The equation for $\Delta
W$ above and its differential form $dW$ provide the correct identification at
the microscopic level of the SI-quantities, and must be used to account for irreversibility.

Furthermore, using Eq. (\ref{E_k-partition}) in Eq. (\ref{Energy Differential}%
), we have%
\begin{align}
d_{\text{i}}E  &  \doteq%
%TCIMACRO{\tsum \nolimits_{k}}%
%BeginExpansion
{\textstyle\sum\nolimits_{k}}
%EndExpansion
p_{k}d_{\text{i}}E_{k}+%
%TCIMACRO{\tsum \nolimits_{k}}%
%BeginExpansion
{\textstyle\sum\nolimits_{k}}
%EndExpansion
E_{k}d_{\text{i}}p_{k}=0,\label{Av-diE}\\
dE  &  =d_{\text{e}}E\doteq%
%TCIMACRO{\tsum \nolimits_{k}}%
%BeginExpansion
{\textstyle\sum\nolimits_{k}}
%EndExpansion
p_{k}d_{\text{e}}E_{k}+%
%TCIMACRO{\tsum \nolimits_{k}}%
%BeginExpansion
{\textstyle\sum\nolimits_{k}}
%EndExpansion
E_{k}d_{\text{e}}p_{k}, \label{Av-deE}%
\end{align}
where we have used the identity$\ d_{\text{i}}W=d_{\text{i}}Q$ from Eq.
(\ref{diQ-diW-EQ}) in the top equation to show consistency of the above
approach with the important identity in Eq. (\ref{diE-EQ}); the first term
here represents $(-d_{\text{i}}W)$ and the second term stands for
$d_{\text{i}}Q$.

\begin{claim}
\label{Claim-diEk-diE} Most important conclusion of our approach is that even
if $d_{\text{i}}E_{k}\neq0$, $d_{\text{i}}E=0$ as is well known; see Eq.
(\ref{diE-EQ}).
\end{claim}

The Eq. (\ref{Av-deE}) reproduces Eq. (\ref{FirstLaw-MI}), which again shows
consistency of our approach with thermodynamics.\ 

\section{Exclusive and Inclusive
Hamiltonians\label{Sec-Exclusive-Inclusive-Hamiltonians}}

The Hamiltonian $\mathcal{H}(V,\xi)$ above is given in terms of SI-variables
and has no information about the medium. Because of this, it has become
customary to call it the exclusive Hamiltonian. However, as $V$ is coupled to
the external pressure $P_{0}$, it is also useful to introduce another
Hamiltonian in which $V$ is replaced by $P_{0}$. This is a well-known trick in
classical equilibrium thermodynamics ($P=P_{0}$), where such a transformation
is known as the Legendre transform. Instead of considering the energy $E(V)$,
we consider its Legendre transform $E(V)+P_{0}V=E(V)+PV$. It is well known
that in equilibrium, this transform is nothing but the enthalpy $H(P)$
\cite{Landau}. A simple example is to consider the situation in Fig.
\ref{Fig_Piston-Spring}(a) in which the system $\Sigma$ is formed by the gas
in the left chamber delimited by the cylinder and the right chamber with the
movable piston and the cylinder is connected to a large sealed container with
a gas at pressure $P_{0}\neq P$; the right chamber and the sealed container
along with the the cylinder and the piston forms $\widetilde{\Sigma}$.

Taking this cue of the EQ Legendre transform, the inclusive Hamiltonian
$\mathcal{H}^{\prime}$\ is defined as \cite{Bochkov,Jarzynski}
\begin{equation}
\mathcal{H}^{\prime}(V,\xi,P_{0})=\mathcal{H}(V,\xi)\mathcal{+}P_{0}V,
\label{InclusiveHamiltonian}%
\end{equation}
even when we are dealing with a NEQ situation such as when $P\neq P_{0}$
and/or $A\neq0$; here%
\[
P_{0}\doteq-\partial\widetilde{E}/\partial\widetilde{V}%
\]
is the conjugate field of the medium; compare this with $P\doteq-\partial
E/\partial V$. As we have shown \cite{Gujrati-II,Gujrati-III}, the NEQ
Legendre transform has a very different property. We find that
\begin{equation}
V=\frac{\partial\mathcal{H}^{\prime}}{\partial P_{0}},A=-\frac{\partial
\mathcal{H}^{\prime}}{\partial\xi},P^{\prime}\doteq-\frac{\partial
\mathcal{H}^{\prime}}{\partial V}=P-P_{0}\neq0 \label{Inclusive-Fields}%
\end{equation}
so that $V,\xi$ and $P_{0}$ are parameters in $\mathcal{H}^{\prime}$ unless
$P^{\prime}$ vanishes, which will happen only under mechanical equilibrium.
Therefore, we can think of $\mathcal{H}^{\prime}(V,\xi,P_{0})$ as another
Hamiltonian with the three parameters $V,\xi$ and $P_{0}$; each parameter will
have its own contribution to work $dW\ $for the inclusive Hamiltonian
$\mathcal{H}^{\prime}(V,\xi,P_{0})$. However, the main difference is that
$P_{0}$ is not an extensive parameter. As it will be treated as a work
variable, the conjugate field will be an extensive variable. The change
$dE^{\prime}$in the inclusive Hamiltonian $\mathcal{H}^{\prime}$ is%
\begin{equation}
dE^{\prime}=-P^{\prime}dV-Ad\xi+VdP_{0}, \label{Inclusive Energy}%
\end{equation}
which will reduce to the EQ form $dE^{\prime}=VdP_{0}$ for the EQ enthalpy. As
has already been discussed in the literature \cite{Bochkov,Jarzynski} the
inclusive Hamiltonian $\mathcal{H}^{\prime}$ should be thought of as referring
to a different system $\Sigma^{\prime}$. This means that $\mathcal{H}^{\prime
}$ or the corresponding energy $E^{\prime}$\ is an\ SI quantity for
$\Sigma^{\prime}$. Thus, there are SI analogs of generalized work $dW^{\prime
}$, generalized heat $dQ^{\prime}$, etc. along with MI analogs like
$d\widetilde{W^{\prime}}$, etc. Indeed, there is \emph{no} fundamental
difference between the exclusive and inclusive approaches. All results for the
inclusive Hamiltonian can be simply converted for the inclusive Hamiltonian by
simply adding a prime on all the quantities and adding the contribution from
the parameter $P_{0}$. This will become clear as we go along.

We see from Eqs. (\ref{Inclusive Energy}) and (\ref{Gibbs-E-IEQ}) that
\begin{equation}
d(E^{\prime}-E)=d(P_{0}V). \label{Incl-Excl-H-Difference-Infinite}%
\end{equation}
This is basically what we see from the definition of $\mathcal{H}^{\prime}$ in
Eq. (\ref{InclusiveHamiltonian}): the difference $\mathcal{H}^{\prime
}-\mathcal{H}$ is nothing but $P_{0}V$. This remains true also for the
accumulated change

\textbf{Comparison with Jarzynski's Approach}: Finally, \ we remark that
Jarzynski's discussion \cite{Jarzynski} of the inclusive Hamiltonian differs
from our discussion in that Jarzynski overlooks the first term $-P^{\prime}dV$
in $dE^{\prime}$ above. Accordingly, he assumes the thermodynamic force
$P^{\prime}=0$, which makes his conclusion very different from us; see
Conclusion \ref{Conclusion-ThermodynamicForce-Irreversibility}. As this is an
important difference, we state it as a

\begin{conclusion}
\label{Conclusion-Jarzynski-InclusiveApproach0}Jarzynski's approach to the
inclusive Hamiltonian sets $P^{\prime}=P-P_{0}=0$. This results in the
complete absence of dissipation as he does not consider any internal variable
$\xi$ in his analysis.
\end{conclusion}

\section{Some Clarifying Examples\label{Sec-Examples}}

Before proceeding further, we clarify the distinction between various
thermodynamic works $dW,d_{\text{e}}W$ and $d\widetilde{W}$ or their
microstate analogs by two simple examples. It is also clear from the previous
section that the pressure difference $P^{\prime}$ plays an important role in
capturing dissipation; see Conclusion
\ref{Conclusion-Jarzynski-InclusiveApproach0}. Only under mechanical
equilibrium do we have $P^{\prime}=0$. The following examples will make this
point abundantly clear that a nonzero force imbalance like $P^{\prime}$ is
just as common even in classical mechanics whenever there is absence of
mechanical equilibrium.

\subsection{Force Imbalance in a Mechanical System: A Microstate
Approach\label{Sec-Example-MechSystem}}

\subsubsection{Exclusive Approach}

Consider as our system a general but a purely classical mechanical
one-dimensional massless spring of arbitrary \emph{exclusive} Hamiltonian
$\mathcal{H}(x)$ with one end fixed at an immobile wall on the left and the
other end with a mass $m$ free to move; see Fig. \ref{Fig_Piston-Spring}(b).
The center of mass of $m$ is located at $x$ from the left wall. For the moment
we consider a vacuum instead of a fluid filling system so we do not need to
worry about any frictional drag; we will consider this complication in Sec.
\ref{Sec-Friction}. The free end is pulled by an \emph{external} force (not
necessarily a constant) $F_{0}$ applied at time $t=0$; thus $x$ acts as a work
parameter. We do not show the center-of-mass momentum $p$ as it plays no role
in determining work. We treat the system purely mechanically. Therefore, the
exercise here should be considered as discussing a microstate of the system.

Initially the spring is undisturbed and has zero\emph{ }SI restoring spring
force
\[
F_{\text{s}}=-\partial\mathcal{H}/\partial x.
\]
The total force
\[
F_{\text{t}}=F_{0}+F_{\text{s}}%
\]
is the force \emph{imbalance} $F_{\text{t}}\lessgtr0$. There is no mechanical
equilibrium unless $F_{\text{t}}=0$ and the spring continues to stretch or
contract, thereby giving rise to an \emph{oscillatory} motion that will go on
for ever. During each oscillation, $F_{\text{t}}$\ is almost always nonzero.
\ The\emph{ }SI work done by $F_{\text{s}}$ is the spring\emph{ }work
$dW\doteq F_{\text{s}}dx$ performed by the spring (internal approach), while
the work performed by $F_{0}$ is $d\widetilde{W}=F_{0}dx$ \emph{transferred}
to the spring; its negative $d_{\text{e}}W=-F_{0}dx$ is the exchange work
(external approach). The kinetic energy plays no role in determining work and
is not considered. Being a purely mechanical example, there is no dissipation.
Despite this, we can introduce using the modern notation
\begin{equation}
d_{\text{i}}W\doteq dW-d_{\text{e}}W\equiv dW+d\widetilde{W}\equiv
F_{\text{t}}dx, \label{diW-microstate}%
\end{equation}
which can be of \emph{either} sign (no second law here) and represents the
work done by $F_{\text{t}}$. Thus, $dW,d_{\text{e}}W=-d\widetilde{W}$ and
$d_{\text{i}}W$ represent \emph{different} works, \emph{a result that has
nothing to do with dissipation but only with the imbalance}; among these, only
the generalized work $dW$ is an SI work. The change in the Hamiltonian
$\mathcal{H}=E$ of the spring due to a variation in the work variable $x$ is
\[
\left.  d\mathcal{H}\right\vert _{\text{w}}=\left.  dE\right\vert _{\text{w}%
}=-Fdx=-dW\neq d\widetilde{W};
\]
the suffix w refers to the change caused by the performing work by varying $x$
here. We thus conclude that%
\begin{equation}
dW=-\left.  dE\right\vert _{\text{w}},d\widetilde{W}=\left.  dE\right\vert
_{\text{w}}+F_{\text{t}}dx, \label{Exclusive-Works}%
\end{equation}
which shows the importance of the force imbalance $F_{\text{t}}$ and also
shows that $dW\neq-d\widetilde{W}$ almost always.

\subsubsection{Inclusive Approach}

Let us consider the \emph{inclusive} Hamiltonian \cite{Jarzynski}
$\mathcal{H}^{\prime}=E^{\prime}\doteq E-F_{0}x$ used in deriving Eq.
(\ref{JarzynskiRelation}); it also explains the prime there. We have
\[
dE^{\prime}=dE-d(F_{0}x)=-F_{\text{t}}dx-xdF_{0}.
\]
As the force $F_{\text{t}}=-\partial E^{\prime}/\partial x$ conjugate to $x$
does not identically vanish, $E^{\prime}(x,F_{0})$ is a function of \emph{two}
work parameters $x$ and $F_{0}$. However, Jarzynski, see Conclusion
\ref{Conclusion-Jarzynski-InclusiveApproach0}, neglects $F_{\text{t}}$ so
$E^{\prime}(F_{0})$ becomes a function of only $F_{0}$
\cite{Gujrati-jarzynski-SecondLaw}. We will see below that the existence of
$F_{\text{t}}$ is very common. As $x=-\partial E^{\prime}/\partial F_{0}$, $x$
is the generalized force conjugate to $F_{0}$. The SI work $dW^{\prime}$
consists of two contributions :
\[
dW^{\prime}=dW_{x}^{\prime}+dW_{F_{0}}^{\prime}=F_{\text{t}}dx+xdF_{0}%
\]
and satisfies
\begin{equation}
dW^{\prime}=-\left.  dE\right\vert _{w}, \label{Inclusive-Work-System}%
\end{equation}
as in the exclusive\emph{ }approach. Furthermore, $dW_{F_{0}}^{\prime}%
=xdF_{0}$ represents the external work
\[
d_{\text{e}}W^{\prime}=dW_{F_{0}}^{\prime}=xdF_{0};
\]
hence,
\[
d_{\text{i}}W^{\prime}=dW_{x}^{\prime}=F_{\text{t}}dx
\]
represents the internal work due to $F_{\text{t}}$, and which appears in the
left side of Eq. (\ref{JarzynskiRelation}).

The following identities are always satisfied, whether we consider a
mechanical system as in this subsection or a thermodynamic system as in the
next subsection:
\begin{subequations}
\label{Inclusive-Exclusive-Change}%
\begin{align}
\left.  dE^{\prime}\right\vert _{\text{w}}-\left.  dE\right\vert _{\text{w}}
&  \doteq dW-dW^{\prime}\equiv-d(F_{0}x),\label{Inclusive-Exclusive-Change-dW}%
\\
d\widetilde{W}^{\prime}-d\widetilde{W}  &  \doteq d_{\text{e}}W-d_{\text{e}%
}W^{\prime}\equiv-d(F_{0}x).\label{Inclusive-Exclusive-Change-dW-tide}\\
d_{\text{i}}W^{\prime}  &  \equiv d_{\text{i}}W\equiv F_{\text{t}}dx.
\label{Inclusive-Exclusive-Change-diW}%
\end{align}
For a mechanical system like a microstate $\mathfrak{m}_{k}$, we should append
a subscript $k$ to each of the quantities in Eq.
(\ref{Inclusive-Exclusive-Change}). For a thermodynamic system, each of the
quantities also refer to thermodynamic average quantities.

Let us investigate the case $F_{0}\equiv0$: $F_{\text{t}}=F_{\text{s}}$, and
$dW=dW^{\prime}=F_{\text{s}}dx\neq0$ and $d\widetilde{W}=d\widetilde
{W}^{\prime}=0$ as a consequence of $\mathcal{H}=\mathcal{H}^{\prime}$. Such a
situation arises when the spring, which is initially kept locked in a
compressed (or elongated) state is unlocked to let go without applying any
external force. Here, $dW=dW^{\prime}\neq0$, while $d\widetilde{W}%
=d\widetilde{W}^{\prime}=0$ as the spring expands (or contracts) under the
influence of its spring force $F_{\text{s}}$. The reader should notice a
similarity with the free expansion noted above.

\textbf{Comparison with Jarzynski's Approach: }The difference between our
approach and that by Jarzynski should be mentioned. As said earlier, Jarzynski
does not allow the contribution $dW_{x}^{\prime}=F_{\text{t}}dx$ in
$dE^{\prime}$ as he treats $E^{\prime}(F_{0})$ as a function of a single work
parameter $F_{0}$. Therefore,
\end{subequations}
\begin{equation}
d_{\text{i}}W^{\prime\text{(J)}}=0\text{ and }d\widetilde{W^{\prime}%
}^{\text{(J)}}=-xdF_{0}=\left.  dE^{\prime}\right\vert _{\text{w}}%
^{\text{(J)}}; \label{JarzynskiWorkEnergy-Mechanical}%
\end{equation}
we have used a superscript (J) as a reminder for the results by Jarzynski.
Thus, as concluded in Sec. \ref{Sec-Exclusive-Inclusive-Hamiltonians},
Jarzynski's approach does not allow any force imbalance $F_{\text{t}}$ (see
Conclusion \ref{Conclusion-Jarzynski-InclusiveApproach0}) so his conclusion
$d\widetilde{W^{\prime}}^{\text{(J)}}=\left.  dE^{\prime}\right\vert
_{\text{w}}^{\text{(J)}}$ fits with the questionable identification
$d\widetilde{W^{\prime}}_{k}\overset{?}{=}\left.  dE_{k}^{\prime}\right\vert
_{\text{w}}$ whose validity requires $d_{\text{i}}W^{\prime}=0$, \textit{i.e.,
}$F_{\text{t}}=0$. As we will see in the next example, $F_{\text{t}}$ is not
only ubiquitous (it is present even in equilibrium) but also necessary for irreversibility.

\subsection{Force Imbalance in a Thermodynamic
System\label{Sec-ForceImbalance-Thermodynamic}}

\subsubsection{Exclusive Approach}

To incorporate dissipation, we consider a thermodynamic analog of the above
example by replacing the vacuum with a fluid in the example discussed above;
see Fig. \ref{Fig_Piston-Spring}(b). This example is no different from the one
shown in Fig. \ref{Fig_Piston-Spring}(a). We will actually discuss this
gas-piston system since we have already studied it earlier. This is easily
converted to the spring-mass system by a suitable change in the vocabulary as
we will elaborate below.

The piston is locked and the gas has a pressure $P$. We first focus on various
work averages to understand the form of dissipation. At time $t=0$, an
external pressure $P_{0}<P$ is applied on the piston and the lock on the
piston is released. We should formally make the substitution $x\rightarrow
V,F\rightarrow P$ (or $P_{k}$ when considering $\mathfrak{m}_{k}$) and
$F_{0}\rightarrow-P_{0}$. At the same time, we will also invoke $\xi$. The gas
expands $(dV\geq0)$ and $P\searrow P_{0}$. The SI work done by the gas is
$dW=PdV+Ad\xi$, while $d\widetilde{W}=-P_{0}dV=-d_{\text{e}}W$. The difference
$d_{\text{i}}W=(P-P_{0})dV+Ad\xi\geq0$ appears as the irreversible work that
is dissipated in the form of heat ($d_{\text{i}}Q\equiv d_{\text{i}}W$) either
due to friction between the moving piston and the cylinder or other
dissipative forces like the viscosity; including them does not change the
first two equations but supplements the meaning of $F_{\text{t}}$\ in the last
equation in Eq. (\ref{Inclusive-Exclusive-Change}); it must include all
possible force imbalance as clearly seen in Sec. \ref{Sec-Friction}. We will
see how internal processes such as friction give rise to internal variables.
Assuming this to be the case, our assumption of the presence of an internal
variable should not come as a surprise.

Let us analyze this model more carefully at a microstate level. Let
$\mathcal{H}(V,\xi)=\allowbreak E(V\xi)$ be the exclusive Hamiltonian of the
gas. Let $E_{k}(V,\xi)$ denote the energy in the inclusive approach of some
$\mathfrak{m}_{k}$; let $E(V,\xi)\doteq\left\langle E\right\rangle $ and
$P(V,\xi)\doteq\left\langle P\right\rangle $ be their statistical average over
microstates. We use Eq. (\ref{E_k-partition}), and use Eq.
(\ref{Works-Microstates}) to identify
\[
dW_{k}\doteq-dE_{k},d_{\text{e}}W_{k}\doteq-d_{\text{e}}E_{k},d_{\text{i}%
}W_{k}\doteq-d_{\text{i}}E_{k},
\]
giving the three work-Hamiltonian relations for a microstate. The statistical
averaging gives
\begin{align*}
dW  &  \equiv-\left.  dE\right\vert _{\text{w}}=PdV+Ad\xi,\\
d_{\text{e}}W  &  \equiv-\left.  d_{\text{e}}E\right\vert _{\text{w}}%
=P_{0}dV\\
d_{\text{i}}W  &  \equiv-\left.  d_{\text{i}}E\right\vert _{\text{w}}%
=(P-P_{0})dV+Ad\xi
\end{align*}
just as discussed above; the suffix w refers to these averages determined by
the work variables $V$ and $\xi$.

There is no reason for $P_{k}=P_{0},A_{k}=0,\forall k$, even in a reversible
process for which $P=P_{0}$ and $A_{0}=0$. As there are pressure fluctuations
even in equilibrium,
\begin{equation}%
%TCIMACRO{\tsum \nolimits_{k}}%
%BeginExpansion
{\textstyle\sum\nolimits_{k}}
%EndExpansion
p_{k}(P_{k}-P_{0})^{2}\geq0. \label{PressureFluctuations}%
\end{equation}
For the same reason, we expect affinity fluctuations also in equilibrium.
Thus, $P_{k}-P_{0}\neq0$ and $A_{k}\neq0$ in general; hence $d_{\text{i}}%
W_{k}\neq0$.

\subsubsection{Inclusive Approach}

The inclusive Hamiltonian for $\mathfrak{m}_{k}$ appears similar to the
microstate analog of the NEQ enthalpy \cite{Gujrati-II,Gujrati-III}
$E_{k}^{\prime}(V,\xi,P_{0},)=E_{k}(V,\xi)+P_{0}V$. It is a function of
$V,\xi$ and $P_{0}$ as was discussed in Sec.
\ref{Sec-Exclusive-Inclusive-Hamiltonians}. Therefore,
\begin{align*}
dW_{k}^{\prime}  &  =-\left(  \partial E_{k}^{\prime}/\partial V\right)
dV-\left(  \partial E_{k}^{\prime}/\partial\xi\right)  d\xi-\left(  \partial
E_{k}^{\prime}/\partial P_{0}\right)  dP_{0}\\
&  =(P_{k}-P_{0})dV+A_{k}d\xi-VdP_{0}\\
&  =P_{k}dV+A_{k}d\xi-d(P_{0}V),
\end{align*}
of this, $d_{\text{e}}W_{k}^{\prime}=-d\widetilde{W}_{k}^{\prime}=-VdP_{0}$ is
spent to overcome the external force $-V\doteq-\left(  \partial E_{k}^{\prime
}/\partial P_{0}\right)  $ conjugate to the work variable $P_{0}$ and the
balance is the irreversible work
\[
d_{\text{i}}W_{k}^{\prime}=(P_{k}-P_{0})dV+A_{k}d\xi=d_{\text{i}}W_{k}.
\]
Since $d\widetilde{W}_{k}=-P_{0}dV$, Eq. (\ref{Inclusive-Exclusive-Change})
remains satisfied for each microstate, and also for the averages.

As we have seen above that $d_{\text{i}}W_{k}\neq0$ in general. Thus, we come
to the very important conclusion

\begin{conclusion}
\label{Conclusion-nonzeo-diW} The presence of $d_{\text{i}}W_{k}\equiv
d_{\text{i}}W_{k}^{\prime}$ is ubiquitous and must be accounted for even in a
reversible process, let alone an irreversible process; see Eq.
(\ref{PressureFluctuations}). This clearly shows that $dW_{k}$ ($dW_{k}%
^{\prime}$) does not vanish even if $d\widetilde{W}_{k}=-d_{\text{i}}W_{k}$
($d\widetilde{W^{\prime}}=-d_{\text{i}}W_{k}^{\prime}$) vanishes.
\end{conclusion}

\begin{conclusion}
\label{Conclusion-dW_k-P_0} The microwork $dW_{k}$ or $dW_{k}^{\prime}$
continues to contribute to $dE_{k}$ (such as during $\overline{\mathcal{P}}$
in the Jarzynsky process) even if $d\widetilde{W}_{k}$ or $d\widetilde
{W^{\prime}}$ has ceased to exist (during $\overline{\mathcal{P}}$).
Therefore, $dW_{k}$ or $dW_{k}^{\prime}$ contributes over the entire process
$\mathcal{P}_{0}$.
\end{conclusion}

\begin{remark}
\label{Remark-DW_k-P_0}The evaluation of $\Delta W_{k}$ or $\Delta
W_{k}^{\prime}$ becomes extremely easy as we need to focus on the entire
process and do not have to consider the driving and reequilibration stages of
the process separately.$\ $As a consequence, we need to determine $\Delta
E_{k}$ between the final and initial states of $\mathcal{P}_{0}$. As $E_{k}$
is an SI-quantity, its value depend on the state, see Definition
\ref{Def-State}, so the value of $\Delta E_{k}$ does not depend on the actual
process $\mathcal{P}_{0}$. In other words, $\Delta W_{k}=-\Delta E_{k}$ is the
same for all possible processes $\mathcal{P}_{0}$ between the same two states.
This observation clarifies a very important Conclusion
\ref{Remark-General Process-DWk}\ obtained later.
\end{remark}

A third example is the spring-mass problem in Fig. \ref{Fig_Piston-Spring}(b),
where we consider the relative motion of the particle. This example can be
considered a prototype of a manipulated Brownian particle undergoing a
relative motion with respect to the rest of the system and is treated within
our internal approach in Sec. \ref{Sec-Friction}. As we will see, the
irreversible work done by the frictional force is properly accounted by the
generalized work and, in particular, the frictional work is part of
$d_{\text{i}}W$ or of $d_{\text{i}}W_{k}$.

\section{Emergence of Internal Variables in the
Hamiltonian\label{Sec-Friction}}

We wish to show that including other dissipation or internal variables does
not alter the first two equations of Eq. (\ref{Inclusive-Exclusive-Change}).
However, the third equation needs to incorporate additional contributions due
to new forms of dissipation such as new internal variables. The discussion
also shows how the Hamiltonian becomes dependent on internal variables, and
how the system is maintained \emph{stationary} despite motion of its parts.

\subsection{Piston-Gas System}

We consider the second example, which is depicted in Fig.
\ref{Fig_Piston-Spring}(a), for this exercise. To describe dissipation, we
need to treat the motion of the piston by including its momentum
$\mathbf{P}_{\text{p}}$ in our discussion. The gas, the cylinder and the
piston constitute the system $\Sigma$. We have a gas of mass $M_{\text{g}}$ in
the cylindrical volume $V_{\text{g}}$, the piston of mass $M_{\text{p}}$, and
the rigid cylinder (with its end opposite to the piston closed) of mass
$M_{\text{c}}$. However, we will consider the composite subsystem
$\Sigma_{\text{gc}}=\Sigma_{\text{g}}\cup\Sigma_{\text{c}}$ so that with
$\Sigma_{\text{p}}$ it makes up $\Sigma$. The Hamiltonian $\mathcal{H}$ of the
system is the sum of $\mathcal{H}_{\text{gc}}$ of the gas and cylinder,
$\mathcal{H}_{\text{p}}$ of the piston, the interaction Hamiltonian
$\mathcal{H}_{\text{int}}$ between the two subsystems $\Sigma_{\text{gc}}%
\ $and $\Sigma_{\text{p}}$, and the stochastic interaction Hamiltonian
$\mathcal{H}_{\text{stoc}}$ between $\Sigma$ and $\widetilde{\Sigma}$. As is
customary, we will neglect $\mathcal{H}_{\text{stoc}}$ here. We assume that
the centers-of-mass of $\Sigma_{\text{gc}}$ and $\Sigma_{\text{p}}$ are moving
with respect to the medium with linear momentum $\mathbf{P}_{\text{gc}}$ and
$\mathbf{P}_{\text{p}}$, respectively. We do not allow any rotation for
simplicity. We assume that
\begin{equation}
\mathbf{P}_{\text{gc}}+\mathbf{P}_{\text{p}}%
=0,\label{Stationary_Momentum_Condition}%
\end{equation}
so that $\Sigma$ is at rest with respect to the medium. Thus,\
\[
\mathcal{H}(\left.  \mathbf{x}\right\vert V,\mathbf{P}_{\text{gc}}%
,\mathbf{P}_{\text{p}})=%
%TCIMACRO{\tsum \nolimits_{\lambda}}%
%BeginExpansion
{\textstyle\sum\nolimits_{\lambda}}
%EndExpansion
\mathcal{H}_{\lambda}(\left.  \mathbf{x}_{\lambda}\right\vert V_{\lambda
},\mathbf{P}_{\lambda})+\mathcal{H}_{\text{int}},
\]
where $\lambda=$gc,p, $\mathbf{x}_{\lambda}\mathbf{=(r}_{\lambda}%
\mathbf{,p}_{\lambda}\mathbf{)}$ denotes a point in the phase space
$\Gamma_{\lambda}$ of $\Sigma_{\lambda}$; $V_{\lambda\text{ }}$is the volume
of $\Sigma_{\mathbf{\lambda}}$, and $V=V_{\text{gc}}+V_{\text{p}}$ is the
volume of $\Sigma$. We do not exhibit the number of particles $N_{\text{g}%
},N_{\text{c}},N_{\text{p}}$ as we keep them fixed. We let $\mathbf{x}$
denotes the collection ($\mathbf{x}_{\text{g}},\mathbf{x}_{\text{c}%
},\mathbf{x}_{\text{p}}$). Thus, the system Hamiltonian $\mathcal{H}(\left.
\mathbf{x}\right\vert V,\mathbf{P}_{\text{gc}},\mathbf{P}_{\text{p}})$ and the
average energy $E$ depend on the parameters $V,\mathbf{P}_{\text{gc}%
},\mathbf{P}_{\text{p}}$. Accordingly, the system entropy, which we assume is
a state function, is written as $S(E,V,\mathbf{P}_{\text{gc}},\mathbf{P}%
_{\text{p}})$. Hence, the corresponding Gibbs fundamental relation becomes
\[
dS=\beta\lbrack dE+PdV-\mathbf{V}_{\text{gc}}\mathbf{\cdot}d\mathbf{P}%
_{\text{gc}}-\mathbf{V}_{\text{p}}\mathbf{\cdot}d\mathbf{P}_{\text{p}}],
\]
where we have used the conventional conjugate fields
\begin{equation}%
\begin{tabular}
[c]{c}%
$\beta\doteq\partial S/\partial E,\beta P\doteq\partial S/\partial V,,$\\
$\beta\mathbf{V}_{\text{gc}}\doteq-\partial S/\partial\mathbf{P}_{\text{gc}%
}\mathbf{,}\beta\mathbf{V}_{\text{p}}\doteq-\partial S/\partial\mathbf{P}%
_{\text{p}}$%
\end{tabular}
\ \ \ \ \ \ \ \ \ \ \label{Conjugate-Fields_Gas-Piston}%
\end{equation}
as shown elsewhere \cite[and references theirin]{Gujrati-II}. Using Eq.
(\ref{Stationary_Momentum_Condition}), we can rewrite this equation as%
\begin{equation}
dS=\beta\lbrack dE+PdV-\mathbf{V\cdot}d\mathbf{P}_{\text{p}}%
]\label{Gibbs-Fundamental_Gas-Piston}%
\end{equation}
in terms of the $\emph{relative}$ \emph{velocity}, also known as the
\emph{drift velocity }$\mathbf{V\doteq V}_{\text{p}}-\mathbf{V}_{\text{gc}}$
of the piston with respect to $\Sigma_{\text{gc}}$. We can cast the drift
velocity term as $\mathbf{V\cdot}d\mathbf{P}_{\text{p}}\equiv\mathbf{F}%
_{\text{p}}\mathbf{\cdot}d\mathbf{R}$, where $\mathbf{F}_{\text{p}}\doteq
d\mathbf{P}_{\text{p}}\mathbf{/}dt$ is the \emph{force}\ and $d\mathbf{R=V}dt$
is the \emph{relative displacement} of the piston. The first law applied to
the stationary $\Sigma$ gives
\begin{equation}
dE=T_{0}d_{\text{e}}S-P_{0}dV\label{First-Law-Piston}%
\end{equation}
in terms of the temperature $T_{0}$ and the pressure $P_{0}$ of the medium.

The internal motions of $\Sigma_{\text{gc}}$\ and $\Sigma_{\text{p}}$\ is not
controlled by any external agent so the relative motion described by the
relative displacement $\mathbf{R}$ represents an \emph{internal variable
}\cite{Kestin} so that the corresponding affinity $\mathbf{F}_{\text{p}0}=0$
for $\widetilde{\Sigma}$. Because of this, Eq. (\ref{First-Law-Piston}) does
not contain the relative displacement $\mathbf{R}$. We now support this claim
using our approach in the following. This also shows how $\mathcal{H}(\left.
\mathbf{x}\right\vert V,\mathbf{P}_{\text{gc}},\mathbf{P}_{\text{p}})$
develops a dependence on the internal variable $\mathbf{R}$. We manipulate
$dS$ in Eq. (\ref{Gibbs-Fundamental_Gas-Piston}) by using the above first law
for $dE$ so that
\[
TdS=T_{0}d_{\text{e}}S+(P-P_{0})dV-\mathbf{F}_{\text{p}}\mathbf{\cdot
}d\mathbf{R,}%
\]
which reduces to%
\[
T_{0}d_{\text{i}}S=(T_{0}-T)dS+(P-P_{0})dV-\mathbf{F}_{\text{p}}\mathbf{\cdot
}d\mathbf{R.}%
\]
This equation expresses the irreversible entropy generation as sum of three
distinct and independent irreversible entropy generations. To comply with the
second law, we conclude that for $T_{0}>0$,%
\begin{equation}
(T_{0}-T)dS\geq0,(P-P_{0})dV\geq0,\mathbf{F}_{\text{p}}\mathbf{\cdot
}d\mathbf{R\leq}0, \label{SecondLaw-Consequences}%
\end{equation}
which shows that each of the components of $d_{\text{i}}S$\ is nonnegative. In
equilibrium, each irreversible component vanishes, which happens when%
\begin{equation}
T\rightarrow T_{0},P\rightarrow P_{0}\text{, and }\mathbf{V}\rightarrow0\text{
or }\mathbf{F}_{\text{p}}\rightarrow0. \label{Equilibrium-Piston}%
\end{equation}
The inequality $\mathbf{F}_{\text{p}}\mathbf{\cdot}d\mathbf{R\leq}0$ shows
that $\mathbf{F}_{\text{p}}$ and $d\mathbf{R}$ are antiparallel, which is what
is expected of a \emph{frictional} force. This causes the piston to finally
come to rest. As $\mathbf{F}_{\text{p}}$ and $\mathbf{V}$ vanish together, we
can express this force as
\begin{equation}
\mathbf{F}_{\text{p}}=-\mu\mathbf{V}f(\mathbf{V}^{2}),
\label{Friction-GeneralForm}%
\end{equation}
where $\mu>0$ and $f$ is an \emph{even} function of $\mathbf{V}$. The medium
$\widetilde{\Sigma}$ is specified by $T=T_{0},P=P_{0}$ and $\mathbf{V}_{0}=0$
or $\mathbf{F}_{\text{p}}=0$. We will take $\mathbf{F}_{\text{p}}$ and
$d\mathbf{R}$ to be colinear and replace $\mathbf{F}_{\text{p}}\mathbf{\cdot
}d\mathbf{R}$ by $-F_{\text{f}}dx$ ($F_{\text{f}}dx\geq0$), where the
magnitude $F_{\text{p}}$ is written as $F_{\text{f}}$\ as a reminder that this
force is responsible for the frictional force and $dx$ is the magnitude of the
relative displacement $d\mathbf{R}$. The sign convention is that $F_{\text{f}%
}$ and increasing $x$ point in the same direction. We can invert Eq.
(\ref{Gibbs-Fundamental_Gas-Piston}) to obtain%
\begin{equation}
dE=TdS-PdV-F_{\text{f}}dx \label{Gibbs-Fundamental-Energy_Gas-Piston}%
\end{equation}
in which $dQ=TdS$ from our general result in Eq. (\ref{ClausiusEquality}).
Comparing the above equation with the first law in Eq. (\ref{FirstLaw-SI}), we
conclude that
\begin{equation}
dW=PdV+F_{\text{f}}dx. \label{Work-Friction}%
\end{equation}
The important point to note is that the friction term $F_{\text{f}}dx$
properly belongs to $dW$.\ As $d_{\text{e}}W=P_{0}dV$, we have%
\begin{equation}
d_{\text{i}}W=(P-P_{0})dV+F_{\text{f}}dx. \label{Irreversible_Work-Piston}%
\end{equation}
Both contributions in $d_{\text{i}}W$ are separately nonnegative. The
corresponding inclusive Hamiltonian is given by $\mathcal{H}^{\prime
}=\mathcal{H}+P_{0}V$. We can easily verify that the first two equations in
Eq. (\ref{Inclusive-Exclusive-Change}) in the main text remain valid with
$x\rightarrow V,F_{0}\rightarrow-P_{0}$ without any modification. The
right-side of the last equation, however, is modified to $\left(
P-P_{0}\right)  dV+F_{\text{f}}dx$ and now contains the internal variable.

We can determine the exchange heat $d_{\text{e}}Q=dQ-d_{\text{i}}W$%
\begin{equation}
d_{\text{e}}Q=TdS-(P-P_{0})dV-F_{\text{f}}dx \label{ExchangeHeat-Friction}%
\end{equation}

We can now consider a microstate $\mathfrak{m}_{k}$. For this we need to
consider $dE_{S}\equiv\left.  dE\right\vert _{\text{w}}$, from which we
determine
\[
dE_{k}=-dW_{k}=P_{k}dV+F_{\text{f}k}dx,
\]
where $F_{\text{f}k}$ is frictional force associated with $\mathfrak{m}_{k}$.
As $d_{\text{e}}E_{k}=-d_{\text{e}}W_{k}=-P_{0}dV$, we also conclude that%
\[
d_{\text{i}}E_{k}=-d_{\text{i}}W=-(P_{k}-P_{0})dV-F_{\text{f}k}dx.
\]
\ 

It should be emphasized that in the above discussion, we have not considered
any other internal motion such as between different parts of the gas besides
the relative motion between $\Sigma_{\text{gc}}$ and $\Sigma_{\text{p}}$.
These internal motions within $\Sigma_{\text{g}}$ can be considered by
following the approach outlined elsewhere \cite{Gujrati-II}. We will not
consider such a complication here.

\subsection{Particle-Spring-Fluid System}

It should be evident that by treating the piston as a mesoscopic particle such
as a pollen or a colloid, we can treat its thermodynamics using the above
procedure. This allows us to finally make a connection with the system
depicted in Fig. \ref{Fig_Piston-Spring}(b) in which the particle (a pollen or
a colloid) is manipulated by an external force $F_{0}$. We need to also
consider two additional forces $F_{\text{s}}$ and $F_{\text{f}}$, both
pointing in the same direction as increasing $x$; the latter is the frictional
force induced by the presence of the fluid in which the particle is moving
around. The analog of Eq. (\ref{Irreversible_Work-Piston}) for this case
becomes%
\begin{equation}
d_{\text{i}}W=(F_{\text{s}}+F_{0})dx+F_{\text{f}}dx=F_{\text{t}}dx,
\label{diW-Particle-Spring}%
\end{equation}
where $F_{\text{t}}=F_{\text{s}}+F_{0}+F_{\text{f}}$. The other two works are
$dW=(F_{\text{s}}+F_{\text{f}})dx$ and $d\widetilde{W}=F_{0}dx=-d_{\text{e}}%
W$. In EQ, $F_{\text{f}}=0$ and $F_{\text{s}}+F_{0}=0$ ($F_{0}\neq0$) to
ensure $d_{\text{i}}W=0$. In this case, $d\widetilde{W}=-dW=F_{0}dx=dE_{S}$,
but this will not be true for a NEQ state.

\subsection{Particle-Fluid System}

In the absence of a spring in the previous subsection, we must set
$F_{\text{s}}=0$ so
\begin{equation}
dW=F_{\text{f}}dx,d\widetilde{W}=F_{0}dx=-d_{\text{e}}W,d_{\text{i}}%
W=(F_{0}+F_{\text{f}})dx. \label{diW-Particle-Fluid}%
\end{equation}
In EQ, $F_{0}+F_{\text{f}}=0$ so that $F_{\text{f}}=-F_{0}$. This means that
in EQ, the particle's nonzero terminal velocity is determined by $F_{0}$ as
expected. In this case, $d\widetilde{W}=-dW=F_{0}dx=dE_{S}$, but this will not
be true for a NEQ state.

\section{Theoretical Consideration\label{Sec-Theory}}

In this section, we allow not just $V$ and $\xi$ as the two work variables in
the exclusive Hamiltonian $\mathcal{H}$ but an arbitrary number of work
variables $\mathbf{W}$, out of which we single out one special work variable
$V$ that is used to define the exclusive Hamiltonian $\mathcal{H}^{\prime
}=\mathcal{H}+P_{0}V$. (One can construct NEQ Legendre transform with more
than one variable, but we will not consider such a complication here.) The
lesson from Conclusion \ref{Conclusion-nonzeo-diW} is that $d_{\text{i}}%
W_{k}=d_{\text{i}}W_{k}^{\prime}$ does not normally vanish. It is this
Conclusion that forms a central core of the theoretical development in this
section along with the concept of SI-quantities. The importance of both these
ideas have not been appreciated to the best of our knowledge but will play a
central role in our discussion below.

\subsection{Thermodynamic Work-Energy Principle}

We first consider the exclusive Hamiltonian $\mathcal{H}$ and consider its
change $\left.  d\mathcal{H}\right\vert _{\mathbf{w}}$ due to work variables
\begin{equation}
\left.  d\mathcal{H}\right\vert _{w}=(\partial\mathcal{H}/\partial
\mathbf{W})\cdot d\mathbf{W}=-\mathbf{F}\cdot d\mathbf{W,}
\label{EnergyChange-W}%
\end{equation}
due to all work variables in $\mathbf{W}$; here%
\[
\mathbf{F\doteq}-\partial\mathcal{H}/\partial\mathbf{W}%
\]
is the generalized force, from which it follows that the SI-work is%
\begin{equation}
dW\doteq\mathbf{F}\cdot d\mathbf{W=-}\left.  d\mathcal{H}\right\vert _{w}.
\label{GeneralWork-Energy-Relation}%
\end{equation}
It is important to recall that $d\mathcal{H}_{\mathbf{W}}$ is the change in
$\mathcal{H}$ at fixed $\mathbf{W}$ and should not be confused with the change
$\left.  d\mathcal{H}\right\vert _{w}$ due to changes in the work variables in
$\mathbf{W}$. Indeed, $\left.  d\mathcal{H}\right\vert _{w}$ $\equiv
d\mathcal{H}_{S}$, which follows from Eq. (\ref{dE_isentropic}), but we use
the current notation $\left.  d\mathcal{H}\right\vert _{w}$ instead to
emphasize the role of work parameters in $\mathbf{W}$.

We first recognize that $\left.  d\mathcal{H}\right\vert _{w}$ represents the
change in a microstate energy $E_{k}$ so we will always think of $\left.
d\mathcal{H}\right\vert _{w}$ as $dE_{k}$ for some $\mathfrak{m}_{k}$.
Following Sec. \ref{Sec-Notation}, we break $\left.  d\mathcal{H}\right\vert
_{w}$\ into two parts, as we did for $dE_{k}$ for $\mathfrak{m}_{k}$ above:
\begin{equation}
\left.  d\mathcal{H}\right\vert _{\text{w}}=\left.  d_{\text{e}}%
\mathcal{H}\right\vert _{\text{w}}+\left.  d_{\text{i}}\mathcal{H}\right\vert
_{\text{w}}, \label{Hamiltonian-Partition}%
\end{equation}
and apply it to some $\mathfrak{m}_{k}$ so that $\left.  d_{\text{e}%
}\mathcal{H}\right\vert _{\text{w}}=d_{\text{e}}E_{k}\doteq-d_{\text{e}}%
W_{k}\equiv d\widetilde{W}_{k}$ is due to the exchange work $d\widetilde
{W}_{k}$ and $\left.  d_{\text{i}}\mathcal{H}\right\vert _{\text{w}%
}=d_{\text{i}}E_{k}\doteq-d_{\text{i}}W_{k}$ is due to the internal work
$d_{\text{i}}W_{k}$ performed by the generalized force $\mathbf{F}$ imbalance
on $\mathfrak{m}_{k}$. We now make the following claim in the form of a

\begin{theorem}
\textbf{ \label{Theorem-Work-Energy-Principle}}$\mathbf{Thermodynamic}$
$\mathbf{Work}$\textbf{-}$\mathbf{Energy}$ $\mathbf{Principle}$ \emph{The
change }$\left.  d\mathcal{H}\right\vert _{w}=dE_{k}$\emph{ in the Hamiltonian
$\mathcal{H}$ due to work only must be identified with the SI-work }$dW_{k}%
$\emph{ and not with }$d\widetilde{W}_{k}$ for $\mathfrak{m}_{k}$; see Eq.
(\ref{GeneralWork-Energy-Relation}). \emph{It} \emph{has two contributions;
see Eq. (\ref{Hamiltonian-Partition}). The first one corresponds to the
external work }$d\widetilde{W}_{k}=-d_{\text{e}}W_{k}$\emph{ \textit{performed
by the medium on }}$\mathfrak{m}_{k}$\ \emph{and the second one to the
negative internal work }$-d_{\text{i}}W_{k}$ \emph{due to the imbalance in the
generalized forces. After\ statistical averaging, }%
\[
\left.  d_{\text{i}}E\right\vert _{\text{w}}\doteq\left\langle \left.
d_{\text{i}}\mathcal{H}\right\vert _{\text{w}}\right\rangle \equiv
-d_{\text{i}}W\doteq-%
%TCIMACRO{\tsum \nolimits_{k}}%
%BeginExpansion
{\textstyle\sum\nolimits_{k}}
%EndExpansion
p_{k}d_{\text{i}}W_{k}\leq0
\]
\emph{ results in dissipation in the system. }
\end{theorem}

\begin{proof}
Based on the three examples, and the general relation in Eq.
(\ref{EnergyChange-W}) or the earlier result in Eq. (\ref{dW-dE}), the proof
for
\[
\left.  d\mathcal{H}\right\vert _{\text{w}}=dE_{k}=-dW_{k}%
\]
is almost trivial. We only have to recognize that both $dE_{k}$ and $dW_{k}$
are SI-quantities as they appear in the identity in Eq. (\ref{EnergyChange-W}%
). The proof here is not restricted to only two work variables. We consider an
arbitrary number of work variables in $\mathbf{W}$. Consider the first law
representation using the MI-quantities in Eq. (\ref{FirstLaw-SI}). The
isentropic change $dE_{S}$ is identified with
\[
dE_{S}=-dW=-\mathbf{F}\cdot d\mathbf{W}.
\]
It represents the average of $-dW_{k}=dE_{k}$, with%
\[
dE_{k}\doteq(\partial E_{k}/\partial\mathbf{W})\cdot d\mathbf{W}%
\equiv\allowbreak d_{\text{e}}E_{k}+d_{\text{i}}E_{k}.
\]
We thus identify $d_{\text{e}}W_{k}=-d_{\text{e}}E_{k}$ and $d_{\text{i}}%
W_{k}=-d_{\text{i}}E_{k}$. This proves the first part. The isometric change in
$dE$ represents the generalized heat $dQ$, as discussed earlier. To proceed
further and better understand the above result and prove the last part, we
turn to the thermodynamics of the system, which we now require to be in
internal equilibrium so that $dQ=TdS$ \cite{Gujrati-Heat-Work,Gujrati-II}. We
rewrite $dE=TdS-dW$ as
\[
dE=T_{0}d_{\text{e}}S-d_{\text{e}}W+T_{0}d_{\text{i}}S+(T-T_{0})dS-d_{\text{i}%
}W,
\]
which leads to
\begin{equation}
T_{0}d_{\text{i}}S=(T_{0}-T)dS+d_{\text{i}}W, \label{diS-diW}%
\end{equation}
which is a different way to write Eq.
(\ref{Irreversible EntropyGeneration-Complete}). Each term on the right side,
being independent of each other, must be nonnegative separately to ensure the
second law in Eq. (\ref{SecondLaw}); compare with Eq.
(\ref{SecondLaw-Consequences}). This proves the last part
\begin{equation}
d_{\text{i}}W\geq0. \label{SecondLaw-IrriversibleWork}%
\end{equation}

\end{proof}

\subsection{Equivalence of Exclusive and Inclusive Approaches}

The discussion above can be extended to the inclusive energy. Indeed, all
relations derived in the exclusive approach have analogs in the inclusive
approach: all we need to do is to insert a prime on each quantity. It should
be stressed that
\[
-(\partial E_{k}/\partial\mathbf{W})\text{ or}-(\partial E_{k}^{\prime
}/\partial\mathbf{W}^{\prime}),\mathbf{W}^{\prime}=(\mathbf{W},P_{0})
\]
represents the "generalized force" and the work has the conventional
form:\ "force"$\times$"distance," contrary to what is commonly stated.
According to the theorem, only SI-works ($-dW_{k},-\Delta W_{k},-dW,-\Delta
W$, and $-\Delta W_{k}^{\prime}$) must be used on the left sides in Eq.
(\ref{MistakenIdentity}) or its analogs and as shown in Eq.
(\ref{Works-Microstates-total}) and its cumulative analog. It is easy to
convince that the Eq. (\ref{Inclusive-Exclusive-Change}) holds in all cases;
see the discussion in Sec. \ref{Sec-Exclusive-Inclusive-Hamiltonians}.

\section{A New NEQ Work Theorem\label{Sec-WorkTheorem}}

In a quantum system, the index $k$ denotes a set of quantum numbers, which we
take to be a \emph{finite} set, \textit{i.e., }having a finite number of
quantum numbers so that the set $\left\{  \mathfrak{m}_{k}\right\}  $ is
countable infinite. Such is the case for a particle in a box treated quantum
mechanically that we investigate in Sec. \ref{Sec-FreeExpansion}. The energy
$E_{k}(\mathbf{W})$ of a microstate $\mathfrak{m}_{k}$ depends on the work
parameter $\mathbf{W}$ so it changes as $\mathbf{W}$ changes during a process
$\mathcal{P}_{0}$ according to Eq. (\ref{EnergyChange-W}) but the set $k$ does
not change. This means that $\mathfrak{m}_{k}$ keeps its identity, while its
energy changes during the work protocol $\mathcal{P}_{0}$. The change of
energy during this protocol is related to the microwork
\begin{equation}
\Delta W_{k}=-(E_{k\text{fin}}-E_{k\text{in}}), \label{Macrowork-Ek}%
\end{equation}
where we have set
\[
E_{k\text{fin}}\doteq E_{k}(\tau),E_{k\text{in}}\doteq E_{k}\left(  0\right)
;
\]
here, $\tau$\ is the duration of the process $\mathcal{P}_{0}$. As we
discussed in reference to Eq. (\ref{Energy Differential}), the change $\Delta
E_{k}=E_{k\text{fin}}-E_{k\text{in}}$ occurs at fixed $p_{k}$ and equals
$-\Delta W_{k}$. Any change in $p_{k}$ requires the microheat $\Delta Q_{k}$,
\textit{i.e.}, either $\Delta_{\text{e}}Q_{k}$ or $\Delta_{\text{i}}Q_{k}$. We
finally conclude that

\begin{conclusion}
\label{Conclusion-Microstate-Evolution-dW} If we are interested in knowing the
cumulative change $\Delta W_{k}$, we only need to determine $\Delta E_{k}$ by
following the same $\mathfrak{m}_{k}$ mechanically during a work protocol
$\mathcal{P}_{0}$. The probability plays no role as $\Delta Q_{k}$ is of no concern.
\end{conclusion}

\begin{remark}
\label{Remark-General Process-DWk}It should be stated here that for all
different processes $\mathcal{P}_{0}(\mathsf{A}_{\text{fin}}\mid
\mathsf{A}_{\text{in}})$'s between the same two states \textsf{A}%
$_{\text{fin}}$ and \textsf{A}$_{\text{in}}$, not necessarily EQ states,
$E_{k\text{fin}}\ $and $E_{k\text{in}}$ are the same, see Definition
\ref{Def-State}, so
\begin{equation}
\Delta W_{k}(\mathcal{P}_{0})=-\Delta E_{k}(\mathsf{A}_{\text{fin}}%
\text{,}\mathsf{A}_{\text{in}}),\text{ }\forall\mathcal{P}_{0}(\mathsf{A}%
_{\text{fin}}\mid\mathsf{A}_{\text{in}})\text{.} \label{Identical DW_k}%
\end{equation}
This allows us to determine the dissipated work%
\[
\Delta_{\text{i}}W_{k}(\mathcal{P}_{0})=-\Delta E_{k}(\mathsf{A}_{\text{fin}%
}\text{,}\mathsf{A}_{\text{in}})-\Delta_{\text{e}}W_{k}(\mathcal{P}_{0})
\]
in which $\Delta_{\text{e}}W_{k}(\mathcal{P}_{0})=-\Delta\widetilde{W}%
_{k}(\mathcal{P}_{0})$ in terms of the external work $\Delta\widetilde{W}%
_{k}(\mathcal{P}_{0})$.
\end{remark}

What the above remark implies is the following: Different processes between
the same two states \textsf{A}$_{\text{fin}}$ and \textsf{A}$_{\text{in}}$
differ not in $\Delta W_{k}(\mathcal{P}_{0})$ but in $\Delta\widetilde{W}%
_{k}(\mathcal{P}_{0})$ or in $\Delta_{\text{i}}W_{k}(\mathcal{P}_{0})$. This
makes microwork unique in that it does not depend on the nature of
$\mathcal{P}_{0}(\mathsf{A}_{\text{fin}}\mid\mathsf{A}_{\text{in}})$. Despite
this, it contains the contribution of dissipation in it given by the average
$\left\langle \Delta_{\text{i}}W_{k}(\mathcal{P}_{0})\right\rangle $.

However, the property of a quantum $\mathfrak{m}_{k}$ maintaining its identity
during $\mathcal{P}_{0}$ is different from the property of a microstate
$\mathfrak{m}_{k}$ in a classical system, for which $\mathfrak{m}_{k}$ is a
small dimensionless volume element $\delta\mathbf{z}_{k}$ in the phase space
$\Gamma$ surrounding a point $\mathbf{z}_{k}$; the collection $\left\{
\delta\mathbf{z}_{k}\right\}  $ covers the entire phase space $\Gamma$.\ This
microstate changes its identity ($\delta\mathbf{z}_{k}\rightarrow
\delta\mathbf{z}_{l},k\neq l$) as it evolves in time following its Hamiltonian
dynamics; recall that $d\mathcal{H}_{\mathbf{W}}=0$ during this dynamics.
Despite this, the evolution is unique so there is a $1$-to-$1$
mapping\ between $\mathfrak{m}_{k}(t)\doteq\delta\mathbf{z}_{k}(t)$ and
$\mathfrak{m}_{k}(t^{\prime})\doteq\delta\mathbf{z}_{l}(t^{\prime})$. This
causes no problem as the change $\Delta W_{k}=-\Delta E_{k}$ is not affected
by the microstate evolution past $\tau$, the duration of $\mathcal{P}_{0}$;
see Conclusions \ref{Conclusion-MicroWorkHeat} and
\ref{Conclusion-Microheat-dE} for more details. In this case, introducing%
\[
E_{k}\left(  \tau\right)  =E(\delta\mathbf{z}_{k}(\tau)),E_{k}\left(
0\right)  =E(\delta\mathbf{z}_{k}(0),
\]
we can write $\Delta W_{k}$ as in Eq. (\ref{Macrowork-Ek}). Thus, whether we
are considering a classical system or a quantum system, we can always express
$\Delta W_{k}$ as in Eq. (\ref{Macrowork-Ek}).

\subsection{Derivation}

Let us evaluate the particular average $\left\langle \cdot\right\rangle _{0}$
in Eq. (\ref{JarzynskiRelation}) but of $e^{\beta_{0}\Delta W_{k}}$ using Eq.
(\ref{Macrowork-Ek}). We first consider the exclusive approach. We have
\begin{align*}
\left\langle e^{\beta_{0}\Delta W}\right\rangle _{0}  &  \doteq%
%TCIMACRO{\tsum \limits_{k}}%
%BeginExpansion
{\textstyle\sum\limits_{k}}
%EndExpansion
\frac{e^{-\beta_{0}E_{k,\text{i}}}}{Z_{\text{in}}(\beta_{0})}e^{\beta
_{0}\Delta W_{k}}\\
&  =%
%TCIMACRO{\tsum \limits_{k}}%
%BeginExpansion
{\textstyle\sum\limits_{k}}
%EndExpansion
\frac{e^{-\beta_{0}E_{k,\text{i}}}}{Z_{\text{in}}(\beta_{0})}e^{--\beta
_{0}(E_{k\text{fin}}-E_{k\text{in}})},
\end{align*}
which leads to%
\[
\left\langle e^{\beta_{0}\Delta W}\right\rangle _{0}=%
%TCIMACRO{\tsum \limits_{k}}%
%BeginExpansion
{\textstyle\sum\limits_{k}}
%EndExpansion
\frac{e^{-\beta_{0}E_{k,\text{fin}}}}{Z_{\text{in}}(\beta_{0})}=\frac
{Z_{\text{fin}}(\beta_{0})}{Z_{\text{in}}(\beta_{0})},
\]
where $Z_{\text{in}}(\beta_{0})$ and $Z_{\text{fin}}(\beta_{0})$ are initial
and final canonical partition functions. Introducing the free energy
difference $\Delta F\doteq F_{\text{fin}}-F_{\text{in}}$, we finally have
\begin{equation}
\left\langle e^{\beta_{0}\Delta W}\right\rangle _{0}=e^{-\beta_{0}\Delta F}.
\label{GujratiRelation}%
\end{equation}
This is our new work theorem involving microworks. The same calculation can be
carried out for the inclusive Hamiltonian, with a similar result except all
the quantities must be replaced by their prime analog as said earlier:%
\begin{equation}
\left\langle e^{\beta_{0}\Delta W^{\prime}}\right\rangle _{0}=e^{-\beta
_{0}\Delta F^{\prime}}. \label{GujratiRelation0}%
\end{equation}
This fixes the original Jarzynski relation in Eq. (\ref{JarzynskiRelation}) by
accounting irreversibility in both approaches, and is valid in all cases,
reversible or otherwise.

\section{The Free Expansion\label{Sec-FreeExpansion}\textsc{ }}

Consider the case of a free expansion ($P_{0}=0$) of a gas in an
\emph{isolated }system of volume $V_{\text{fin}}$, divided by an impenetrable
partition into the left (L) and the right (R) chambers as shown in Fig.
\ref{Fig_Expansion}(a). Initially, all the $N$ particles are in the left
chamber of volume $V_{\text{in}}$ in an equilibrium state at temperature
$T_{0}$; there is vacuum in the right chamber. At time $t=0$, the partition is
suddenly removed, shown by the broken partition in Fig. \ref{Fig_Expansion}(b)
and the gas is allowed to undergo \emph{free expansion} to the final volume
$V_{\text{fin}}$ during $\mathcal{P}$. After the free expansion, the gas is in
a NEQ state and is brought in contact with $\widetilde{\Sigma}_{\text{h}}$
during $\overline{\mathcal{P}}$ to come to equilibrium at the initial
temperature $T_{0}$. If the gas is ideal, there is no need to bring in
$\widetilde{\Sigma}_{\text{h}}$ for reequilibration; we can let the gas come
to equilibrium by itself as it is well known that the temperature of the
equilibrated gas after free expansion is also $T_{0}$.%
%TCIMACRO{\FRAME{ftbpFU}{3.0208in}{2.2303in}{0pt}{\Qcb{Free expansion of a gas.
%The gas is confined to the left chamber, which is separated by a partition
%(shown by a solid black vertical line) from the vacuum as shown in (a). At
%time $t=0$, the partition is removed abruptly as shown by the broken line in
%its original place in (b). The gas expands in the empty space on the right but
%the expansion is gradual as shown by the solid front, which separates it from
%the vacuum on its right. }}{\Qlb{Fig_Expansion}}{FreeExpansionBoxNew.eps}%
%width 3.0208in;  height 2.2303in;  depth 0pt;  original-width 4.2756in;
%original-height 3.1453in;  cropleft "0";  croptop "1";  cropright "1";
%cropbottom "0";  filename 'FreeExpansionBoxNew.eps';file-properties "XNPEU";}}
%}%
%BeginExpansion
\begin{figure}
[ptb]
\begin{center}
\includegraphics[
height=2.2303in,
width=3.0208in
]%
{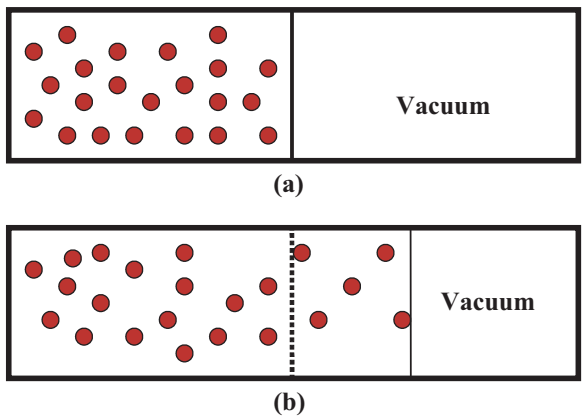}%
\caption{Free expansion of a gas. The gas is confined to the left chamber,
which is separated by a partition (shown by a solid black vertical line) from
the vacuum as shown in (a). At time $t=0$, the partition is removed abruptly
as shown by the broken line in its original place in (b). The gas expands in
the empty space on the right but the expansion is gradual as shown by the
solid front, which separates it from the vacuum on its right. }%
\label{Fig_Expansion}%
\end{center}
\end{figure}
%EndExpansion
It should be stated, which is also evident from Fig. \ref{Fig_Expansion}(b),
that while the removal of the partition can be instantaneous, the actual
process of gas expanding in the right chamber is continuous and gradually
fills it. Therefore, at each instant, it is possible to imagine a front of the
expanding gas shown by the solid vertical line enclosing the largest among
smallest possible volumes containing all the particles so that there are no
particles to the right of it in the right chamber in all possible realization
of the expanding gas. By this we mean the following: we consider all possible
realizations of the expanding gas at a particular time $t>0$ and locate the
front corresponding to the smallest volume containing all the gas particles to
its left. Then we choose among all these fronts that particular front that
results in the smallest volume on its right or the largest volume on its left.
In this sense, this front is an average concept and is shown in Fig.
\ref{Fig_Expansion}(b). We have identified the volume to its right as "vacuum"
in the figure. This means that at each instant when there is a vacuum to the
right of this front, the gas is expanding against zero pressure so that
$d\widetilde{W}=0$. Despite this, as the expansion is a NEQ process,
$dW=d_{\text{i}}W>0$.

\subsection{Quantum Free Expansion}

We now apply Eq. (\ref{GujratiRelation}) to the free expansion of a
one-dimensional \emph{ideal} gas of classical particles, but treated quantum
mechanically as a particle in a box with rigid walls, which has studied
earlier \cite{Gujrati-QuantumHeat}; see also Bender et al \cite{Bender}. We
assume that the gas is thermalized initially at some temperature
$T_{0}=1/\beta_{0}$ and then isolated from the medium so that the free
expansion occurs in an isolated system. After the free expansion from the box
size $L_{\text{in}}$ to $L_{\text{fin}}>L_{\text{in}}$, the box is again
thermalized at the same temperature $T_{0}$. The role of $V$ is played by the
length $L$ of the box. This will also set the stage for the classical
treatment later.

As discussed in Sec. \ref{Sec-Example-MechSystem} for $F_{0}=0$, $\Delta
W_{k}\neq0$ even though $\Delta\widetilde{W}_{k}=0$. Since we are dealing with
an ideal gas, we do not need to bring $\widetilde{\Sigma}_{\text{h}}$, see
below, so we let the gas to come to equilibrium as an isolated system. As
there is no inter-particle interaction, we can focus on a single particle for
our discussion; its energy levels are in appropriate units
\[
E_{k}=k^{2}/L^{2},
\]
where $L$ is the length of the box. We assume that the gas is thermalized
initially at some temperature $T_{0}=1/\beta_{0}$. It is isolated from the
medium so that the free expansion occurs in an isolated system, during which,
we have $\Delta_{\text{e}}Q=\Delta_{\text{e}}W=0$ (but $\Delta_{\text{i}%
}Q=\Delta_{\text{i}}W\neq0$) so that $\Delta E_{\text{free}}(L_{\text{fin}%
},L_{\text{in}})=0$; see Eq. (\ref{FirstLaw-MI}). After the free expansion
from the box size $L_{\text{in}}$ to $L_{\text{fin}}>L_{\text{in}}$, the box
is allowed to come to equilibrium in isolation so that we have $\Delta
E_{\text{reeq}}(L_{\text{fin}})=0$. Accordingly, $\Delta E_{\text{eq}%
}(L_{\text{fin}},L_{\text{in}})=0$ after reequilibration.

The initial partition function is given by
\[
Z_{\text{in}}(\beta_{0},L)=%
%TCIMACRO{\tsum \nolimits_{k}}%
%BeginExpansion
{\textstyle\sum\nolimits_{k}}
%EndExpansion
e^{\beta_{0}E_{k,\text{in}}}.
\]
Approximating the sum by an integration over $k$ as is common, we can evaluate
$Z_{\text{in}}(\beta_{0},L)$ from which we find that the free energy
$F_{\text{eq}}$ and the average energy $E_{\text{eq}}\ $are given by
\[
\beta_{0}F_{\text{eq}}=-(1/2)\ln(L^{2}\pi/4\beta_{0}),E_{\text{eq}}%
=1/2\beta_{0};
\]
while $F_{\text{eq}}$ depends on $\beta_{0}$ and $L$, $E_{\text{eq}}$ depends
only on $\beta_{0}$ but not on $L$ so that $E_{\text{eq}}$ has the same value
in the final EQ state. This means that the final equilibrium state has the
same temperature $T_{0}$. This explains why we did not need to bring
$\widetilde{\Sigma}_{\text{h}}$ in play for reequilibration as assumed above.\ 

As we have discussed in reference to Eq. (\ref{Energy Differential}) and
concluded in Conclusions \ref{Conclusion-MicroWorkHeat},
\ref{Conclusion-Microheat-dE} and \ref{Conclusion-Microstate-Evolution-dW},
and summarized in Remark \ref{Remark-General Process-DWk}, $\Delta E_{k}=-$
$\Delta W_{k}$ regardless of whether $\mathcal{P}_{0}$\ is irreversible or
not. Below we will show that the present calculation here is dealing with an
irreversible $\mathcal{P}_{0}$. The energy change between two EQ states is
\[
\Delta E_{k}=k^{2}(1/L_{\text{fin}}^{2}-1/L_{\text{in}}^{2}).
\]
Let us determine the microwork done to take the initial microstate to the
final microstate by using the internal pressure%
\begin{equation}
P_{k}=-\partial E_{k}/\partial L=2E_{k}/L\neq0 \label{P_k-FreeExpansion}%
\end{equation}
in%
\begin{equation}
\Delta W_{k}=%
%TCIMACRO{\dint \nolimits_{L_{\text{in}}}^{L_{\text{fin}}}}%
%BeginExpansion
{\displaystyle\int\nolimits_{L_{\text{in}}}^{L_{\text{fin}}}}
%EndExpansion
P_{k}dL. \label{DW-FreeExpansion}%
\end{equation}
It is easy to see that this microwork is precisely equal to \ $-\Delta E_{k}$
in accordance with Theorem \ref{Theorem-Work-Energy-Principle}, as expected.
It is also evident from Eq. (\ref{P_k-FreeExpansion}) that for each $L$
between $L_{\text{in}}$ and $L_{\text{fin}}$,
\[
P=%
%TCIMACRO{\tsum \nolimits_{k}}%
%BeginExpansion
{\textstyle\sum\nolimits_{k}}
%EndExpansion
p_{k}P_{k}=2E/L\neq0.
\]
We can use this average pressure to calculate the thermodynamic work
\[
\Delta W=%
%TCIMACRO{\dint \nolimits_{L_{\text{in}}}^{L_{\text{fin}}}}%
%BeginExpansion
{\displaystyle\int\nolimits_{L_{\text{in}}}^{L_{\text{fin}}}}
%EndExpansion
PdL=2%
%TCIMACRO{\tsum \nolimits_{k}}%
%BeginExpansion
{\textstyle\sum\nolimits_{k}}
%EndExpansion%
%TCIMACRO{\dint \nolimits_{L_{\text{in}}}^{L_{\text{fin}}}}%
%BeginExpansion
{\displaystyle\int\nolimits_{L_{\text{in}}}^{L_{\text{fin}}}}
%EndExpansion
p_{k}E_{k}dL/L\neq0.
\]
as expected. As $\Delta E=0$, it means that $\Delta Q=\Delta W\neq0$, which
really means $\Delta_{\text{i}}Q=\Delta_{\text{i}}W\neq0$ in this case. This
establishes that the expansion we are studying is irreversible. This is also
evident from the observation that $P\neq P_{0}=0$.

Despite this, $\Delta W_{k}$ is always equal to the same $\left(  -\Delta
E_{k}\right)  $ regardless of the nature of irreversibility of $\mathcal{P}%
_{0}$, which is consistent with Conclusion
\ref{Conclusion-Microstate-Evolution-dW} and Remark
\ref{Remark-General Process-DWk}. The same $\Delta W_{k}$ will also apply to a
reversible $\mathcal{P}_{0}$ as we are considering the energy change between
two EQ states. The only difference is that now $\Delta Q=\Delta W\neq0$ will
mean $\Delta_{\text{e}}Q=\Delta_{\text{e}}W\neq0$. It is trivially seen that
Eq. (\ref{GujratiRelation}) is satisfied for all $\mathcal{P}_{0}$, not just
the free expansion.

As $P_{0}=0$, there is no difference between the exclusive Hamiltonian and the
inclusive Hamiltonian. Thus, the discussion above is also valid for the
inclusive Hamiltonian and Eq. (\ref{GujratiRelation0}) with $E_{k}^{\prime
}=E_{k}$ and $F^{\prime}=F$.
%TCIMACRO{\FRAME{ftbpFU}{3.4783in}{2.0825in}{0pt}{\Qcb{The evolution of a
%microstate $\QTR{bf}{z}_{0}\in\QTR{bf}{\Gamma}_{\text{in}},\QTR{bf}{z}\in
%\QTR{bf}{\Gamma}_{\text{fin}}\backslash\QTR{bf}{\Gamma}_{\text{in}}$ following
%microwork (green arrows) into $\QTR{bf}{z}_{\gamma}$ and $\QTR{bf}{z}_{\gamma
%}^{\prime}$, respectively. The initial and final phase spaces are
%$\QTR{bf}{\Gamma}_{\text{in}}$ and $\QTR{bf}{\Gamma}_{\text{fin}}$ shown by
%the interiors of the red ellipses.}}{\Qlb{Fig_PhasePoint-Evolution}%
%}{PhasePoint-Evolution-Reduced.eps}{\special{ language "Scientific Word";
%type "GRAPHIC";  maintain-aspect-ratio TRUE;  display "USEDEF";
%valid_file "F";  width 3.4783in;  height 2.0825in;  depth 0pt;
%original-width 5.3082in;  original-height 3.1609in;  cropleft "0";
%filename 'PhasePoint-Evolution-Reduced.eps';file-properties "XNPEU";}} }%
%BeginExpansion
\begin{figure}
[ptb]
\begin{center}
\includegraphics[
height=2.0825in,
width=3.4783in
]%
{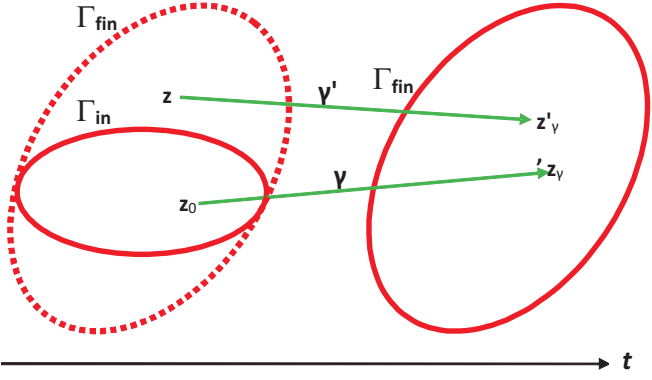}%
\caption{The evolution of a microstate $\mathbf{z}_{0}\in\mathbf{\Gamma
}_{\text{in}},\mathbf{z}\in\mathbf{\Gamma}_{\text{fin}}\backslash
\mathbf{\Gamma}_{\text{in}}$ following microwork (green arrows) into
$\mathbf{z}_{\gamma}$ and $\mathbf{z}_{\gamma}^{\prime}$, respectively. The
initial and final phase spaces are $\mathbf{\Gamma}_{\text{in}}$ and
$\mathbf{\Gamma}_{\text{fin}}$ shown by the interiors of the red ellipses.}%
\label{Fig_PhasePoint-Evolution}%
\end{center}
\end{figure}
%EndExpansion

\subsection{Classical Free Expansion}

We will consider the free expansion of an isolated classical gas in a vacuum
($P_{0}=0$), see Fig. \ref{Fig_Expansion}. We set $V_{\text{in}}$ and
$V_{\text{fin}}$ for simplicity The initial phase space is denoted by the
interior of the solid red ellipse $\boldsymbol{\Gamma}_{\text{in}}$ on the
left side in Fig. \ref{Fig_PhasePoint-Evolution}. The final phase space is
shown by the interior of the broken red ellipse on the left and the solid red
ellipse $\boldsymbol{\Gamma}_{\text{fin}}$ on the right in Fig.
\ref{Fig_PhasePoint-Evolution}. The gas is in a "restricted (\textit{i.e.,
}being confined in the left chamber)" equilibrium state with equilibrium
microstate probability%
\begin{equation}
f_{0}(\delta\mathbf{z}_{0})=e^{-\beta_{0}E(\mathbf{z}_{0})}/Z_{\text{in}%
}(\beta_{0},V_{\text{in}}) \label{probability-dist-initial}%
\end{equation}
at $t=0$; here, the initial partition function in the initial volume
$V_{\text{in}}$ is%
\begin{equation}
Z_{\text{in}}(\beta_{0},V_{\text{in}})\doteq%
%TCIMACRO{\tsum \limits_{\delta\mathbf{z}_{0}\in\boldsymbol{\Gamma}_{\text{in}%
%}}}%
%BeginExpansion
{\textstyle\sum\limits_{\delta\mathbf{z}_{0}\in\boldsymbol{\Gamma}_{\text{in}%
}}}
%EndExpansion
e^{-\beta_{0}E(\mathbf{z}_{0})}. \label{PF-initial}%
\end{equation}
We consider the set of microstates in the final phase space
$\boldsymbol{\Gamma}_{\text{fin}}$ and pick two microstates $\delta
\mathbf{z}_{0}$ and $\delta\mathbf{z}$ associated with $\mathbf{z}_{0}%
\in\boldsymbol{\Gamma}_{\text{in}}$ and $\mathbf{z}\in\overline
{\boldsymbol{\Gamma}}\boldsymbol{\doteq\Gamma}_{\text{fin}}\backslash
\boldsymbol{\Gamma}_{\text{in}}$; here, $\overline{\boldsymbol{\Gamma}%
}\boldsymbol{\doteq\Gamma}_{\text{fin}}\backslash\boldsymbol{\Gamma
}_{\text{in}}$ denotes the difference set of $\boldsymbol{\Gamma}_{\text{fin}%
}$ and $\boldsymbol{\Gamma}_{\text{in}}$. We use the notation $\overline
{\mathbf{z}}_{0}\doteq(\mathbf{z}_{0},\mathbf{z})$ to denote the two points.
Let us identify $(\mathbf{z}_{\gamma},\mathbf{z}_{\gamma}^{\prime})$ as the
\emph{unique} $1$-to-$1$ phase points obtained by the deterministic
Hamiltonian evolution of $(\mathbf{z}_{0},\mathbf{z})$ along the deterministic
or mechanical trajectories $\gamma=\gamma(\mathbf{z}_{0})$ and $\gamma
^{\prime}=\gamma^{\prime}(\mathbf{z})$\ corresponding to a given work protocol
$\mathcal{P}_{0}$; see Fig. \ref{Fig_PhasePoint-Evolution}. The probabilities
of the two paths are irrelevant for the microworks%
\begin{equation}%
\begin{array}
[c]{c}%
\Delta W_{\gamma}(\mathbf{z}_{0})=-(E(\mathbf{z}_{\gamma})-E(\mathbf{z}%
_{0})),\\
\Delta W_{\gamma^{\prime}}(\mathbf{z})=-(E(\mathbf{z}_{\gamma}^{\prime
})-E(\mathbf{z}));
\end{array}
\label{MicroWork-PhaseSpace}%
\end{equation}
see Conclusions \ref{Conclusion-MicroWorkHeat} and
\ref{Conclusion-Microheat-dE}.

While the initial EQ probability distribution $f_{0}(\delta\mathbf{z}_{0})$ is
nonzero for $\delta\mathbf{z}_{0}\in\boldsymbol{\Gamma}_{\text{in}}$, it is
common to think of $f_{0}(\delta\mathbf{z})=0$ for $\mathbf{z}\in
\overline{\boldsymbol{\Gamma}}$. This is an ideal situation and requires
taking the energy $E(\mathbf{z})=\infty$, but in reality, $f_{0}%
(\delta\mathbf{z})$ falls rapidly as we move into the right chamber away from
the left one in the initial macrostate. Moreover, during free expansion,
$f(\delta\mathbf{z})$ at $t>0$ is not going to remain zero. Therefore, we
\emph{formally assume} that the initial probability distribution $f_{0}%
(\delta\mathbf{z})$ is infinitesimally small by assigning to it a very large
positive energy
\begin{equation}
E(\mathbf{z})=e(\mathbf{z})/\varepsilon>0,\mathbf{z}\in\overline
{\boldsymbol{\Gamma}}\text{ at }t=0 \label{Energy-Infinite}%
\end{equation}
by introducing an infinitesimal positive quantity $\varepsilon$. At the end of
the calculation, we will take the limit $\varepsilon\rightarrow0^{+}$, which
simply means $\varepsilon\rightarrow0$ from the positive site. Under this
limit, the contribution from $e^{-\beta_{0}E(\mathbf{z})}$ will vanish:%
\[
e^{-\beta_{0}E(\mathbf{z})}\overset{\varepsilon\rightarrow0^{+}}{\rightarrow
}0.
\]
This allows us to recast the initial partition function as a sum over all
microstates $\overline{\mathbf{z}}\in\boldsymbol{\Gamma}_{\text{fin}}$:%
\begin{equation}
\lim_{\varepsilon\rightarrow0^{+}}Z_{\text{in}}^{\prime}(\beta_{0}%
,V_{\text{fin}},\varepsilon)\doteq\lim_{\varepsilon\rightarrow0^{+}}%
%TCIMACRO{\tsum \limits_{\delta\overline{\mathbf{z}}\in\boldsymbol{\Gamma
%}_{\text{fin}}}}%
%BeginExpansion
{\textstyle\sum\limits_{\delta\overline{\mathbf{z}}\in\boldsymbol{\Gamma
}_{\text{fin}}}}
%EndExpansion
e^{-\beta_{0}E(\overline{\mathbf{z}})}=Z_{\text{in}}(\beta_{0},V_{\text{in}});
\label{Limit-PF}%
\end{equation}
Thus, we can focus on $\boldsymbol{\Gamma}_{\text{fin}}$ as the phase space to
consider during any work protocol $\mathcal{P}_{0}$ instead of
$\boldsymbol{\Gamma}_{\text{in}}$. This allows us to basically use a
$1$-to-$1$ mapping between initial microstates $\overline{\mathbf{z}}%
_{0}\doteq(\mathbf{z}_{0},\mathbf{z})$ and final microstates $\overline
{\mathbf{z}}_{\gamma}\doteq(\mathbf{z}_{\gamma},\mathbf{z}_{\gamma}^{\prime}%
)$\textbf{ }discussed above.

We simply denote $\mathbf{z}_{0}$ or $\mathbf{z}$ by $\overline{\mathbf{z}}%
\in\boldsymbol{\Gamma}_{\text{fin}}$\ or $\overline{\mathbf{z}}_{\gamma}%
\in\boldsymbol{\Gamma}_{\text{fin}}$ for the Hamiltonian evolution of
$\overline{\mathbf{z}}$ along the microwork protocol from now on. We consider
the Jarzynski average of the exponential work in Eq. (\ref{GujratiRelation})
for the exclusive Hamiltonian and write it as%
\begin{equation}
\lim_{\varepsilon\rightarrow0^{+}}\left\langle e^{\beta_{0}\Delta
W}\right\rangle _{0}=\lim_{\varepsilon\rightarrow0^{+}}\frac{%
%TCIMACRO{\tsum \limits_{\delta\overline{\mathbf{z}}\in\boldsymbol{\Gamma
%}_{\text{fin}}}}%
%BeginExpansion
{\textstyle\sum\limits_{\delta\overline{\mathbf{z}}\in\boldsymbol{\Gamma
}_{\text{fin}}}}
%EndExpansion
e^{-\beta_{0}E(\overline{\mathbf{z}})}e^{-\beta_{0}[E(\overline{\mathbf{z}%
}_{\gamma})-E(\overline{\mathbf{z}})]}}{Z_{\text{in}}(\beta_{0},V_{\text{fin}%
},\varepsilon)}, \label{AverageExpWork2}%
\end{equation}
where we have used $\Delta W_{\gamma}(\overline{\overline{\mathbf{z}}%
})=-(E(\overline{\overline{\mathbf{z}}}_{\gamma})-E(\overline{\overline
{\mathbf{z}}}))$ in accordance with Eq. (\ref{MicroWork-PhaseSpace}). The
initial partition function in the original volume $V_{\text{in}}$ because of
the vanishing probabilities to be outside this volume. Because of the
$1$-to-$1$ mapping to $\overline{\mathbf{z}}_{\gamma}$, we can replace the sum
to a sum over $\overline{\mathbf{z}}_{\gamma}$, and at the same time cancel
the initial energy $E(\overline{\mathbf{z}})$ in the exponent; the
cancellation is \emph{exact} even for $\overline{\mathbf{z}}=\mathbf{z}$ for
which $E(\mathbf{z})\rightarrow+\infty$ in the limit $\varepsilon
\rightarrow0^{+}$. Because of this, the $\lim$ operation has no effect in the
numerator. The partition function in denominator reduces to $Z_{\text{in}%
}(\beta_{0},V_{\text{in}})$ as shown in Eq. (\ref{Limit-PF}). We finally find%
\begin{equation}
\left\langle e^{\beta_{0}\Delta W}\right\rangle _{0}=\frac{%
%TCIMACRO{\tsum \limits_{\delta\mathbf{z}_{\gamma}\in\boldsymbol{\Gamma
%}_{\text{fin}}}}%
%BeginExpansion
{\textstyle\sum\limits_{\delta\mathbf{z}_{\gamma}\in\boldsymbol{\Gamma
}_{\text{fin}}}}
%EndExpansion
e^{-\beta_{0}E(\mathbf{z}_{\gamma})}}{Z_{\text{in}}(\beta_{0},V_{\text{in}}%
)}=\frac{Z_{\text{fin}}(\beta_{0},V_{\text{fin}})}{Z_{\text{in}}(\beta
_{0},V_{\text{in}})}, \label{AverageExpWork3}%
\end{equation}
which is precisely what we wish to prove in Eq. (\ref{GujratiRelation}).

The situation with the inclusive Hamiltonian is the same as $\mathcal{H}%
^{\prime}=\mathcal{H}$ as before. This allows us to also prove Eq.
(\ref{GujratiRelation0}). Moreover, as said in the previous subsection, the
demonstration of Eqs. (\ref{GujratiRelation}) and (\ref{GujratiRelation0}) is
valid for any arbitrary process, not just the free expansion, the title of
this section.

It should be emphasized that allowing for a negligible probability is a common
practice even in EQ statistical mechanics where we evaluate the partition
function by considering all microstate, regardless of how negligibly small the
corresponding statistical weight is. This probability could even be zero. The
only difference is that the microstate is defined over the volume of the
system and not outside. We have allowed microstates in deriving Eqs.
(\ref{GujratiRelation}) and (\ref{GujratiRelation0}) with vanishing small or
zero probabilities. Here, we are considering microstates outside the volume of
the system, but mathematically, there is no difference.

By allowing such microstates in $\overline{\boldsymbol{\Gamma}}$, we have
shown that Eqs. (\ref{GujratiRelation}) and (\ref{GujratiRelation0}) hold even
for free expansion of a classical or quantum gas, where the Jarzynski equality fails.

\section{Conclusion\label{Sec-Conclusion}\textsc{ }}

The present work is motivated by a desire to clarify the connection between
work and energy change at the mechanical level. The goal is summarized in
Proposition \ref{Proposition-Goal}. The use of SI-quantities proves useful in
expressing the First law in which the generalized heat $dQ$ is proportional to
$dS$, while the generalized work $dW=-dE_{S}$ is isentropic change in the
energy $E$. This ensures that $dW$ and $dQ$ change $E$ independently so that
we need not consider any effect of $dQ$ while considering the change $dE_{S}$;
see Conclusion \ref{Conc-Isentropic-Isometric-Changes}.

As the probability during microwork does not change, a microstate can also be
treated as a mechanical system during microwork; see Conclusion
\ref{Conclusion-MicroWorkHeat}. As the change $\Delta E_{k}$ is an
SI-quantity, it must be related to the work that is also an SI-quantity. This
gives $\Delta E_{k}$ $=-\Delta_{k}W$ (internal scheme), which contradicts the
current but erroneous practice of using $\Delta E_{k}$ $\overset{?}{=}%
\Delta_{k}\widetilde{W}$ (external scheme) in diverse applications in
nonequilibrium statistical thermodynamics mentioned earlier; see the
discussion in Sec. \ref{Sec-MistakenIdentity}. The mistake results in setting
\begin{equation}
\Delta_{\text{i}}W_{k}=0, \label{IrreversibleWork}%
\end{equation}
which contradicts the most important result of this work that $\Delta
_{\text{i}}W_{k}\neq0$ in almost all cases as concluded in Conclusion
\ref{Conclusion-nonzeo-diW}. This has been justified by a purely mechanical
system in Sec. \ref{Sec-Example-MechSystem}, which shows that Eq.
(\ref{IrreversibleWork}) holds only when there is mechanical equilibrium. This
equilibrium occurs under a very special condition of zero force imbalance
$F_{\text{t}}$. However, as shown in Sec. \ref{Sec-Example-MechSystem}, Eq.
(\ref{IrreversibleWork}) does not hold in general for a mechanical system. The
thermodynamic system investigated in Sec.
\ref{Sec-ForceImbalance-Thermodynamic} also shows that Eq.
(\ref{IrreversibleWork}) does not hold in general.

As Conclusion \ref{Conclusion-ForceImbalance-Dissipation} shows, the presence
a \emph{nonzero force imbalance} is \emph{necessary} (but not sufficient) for
dissipation in the system. The force imbalance is what gives rise to
thermodynamic forces, whose importance does not seem to have been acknowledge
by many workers that consistently use the mistaken identity $\Delta E_{k}$
$=\Delta_{k}\widetilde{W}$ in Eq. (\ref{MistakenIdentity}). This means that
Eq. (\ref{IrreversibleWork}) holds in these studies, which cannot include any
dissipation. This is consistent with Theorem
\ref{Theorem-Work-Energy-Principle} that there is no irreversibility
($\Delta_{\text{i}}W=0$) if we accept $\Delta E_{k}$ $=\Delta_{k}\widetilde
{W}$. This also shows that the JE is not a valid identity.

The existence of $\Delta_{\text{i}}E_{k}=-\Delta_{\text{i}}W_{k}$ is one of
the most surprising result of our approach, which appears almost
counter-intuitive and has remain hitherto unrecognized in the field because of
it. It is so because it is well known that it is impossible to have
$\Delta_{\text{i}}E\neq0$, where
\[
\Delta_{\text{i}}E=%
%TCIMACRO{\tint \nolimits_{\gamma}}%
%BeginExpansion
{\textstyle\int\nolimits_{\gamma}}
%EndExpansion
d_{\text{i}}E,
\]
in terms of $d_{\text{i}}E$ is given in Eq. (\ref{Av-diE}). However, we have
shown, see Claim \ref{Claim-diEk-diE} that even if $d_{\text{i}}E_{k}\neq0$,
$d_{\text{i}}E$ always vanishes.

As we see from Remarks \ref{Remark-DW_k-P_0} and
\ref{Remark-General Process-DWk}, the determination of the SI-work $\Delta
W_{k}$ or $\Delta W_{k}^{\prime}$ over\ the entire process $\mathcal{P}_{0}$
is extremely easy as it is independent of the nature of the actual process
$\mathcal{P}_{0}$; it is the same for all processes between the same two
states since $\Delta E_{k}$ or $\Delta E_{k}^{\prime}$\ is an SI-quantity and
is determined by the initial and the final states but not on the nature of the
process. Because of this, Eqs. (\ref{GujratiRelation}) and
(\ref{GujratiRelation0}) hold for all processes between the same two states.
This is very different from $\Delta_{\text{e}}W_{k}=-\Delta_{k}\widetilde{W}$
or $\Delta_{\text{e}}W_{k}^{\prime}=-\Delta_{k}\widetilde{W^{\prime}}$, which
depends strongly on $\mathcal{P}$ during which it is nonzero and vanishes
during $\overline{\mathcal{P}}$. The difference is $\Delta_{\text{i}}%
W_{k}=\Delta_{\text{i}}W_{k}^{\prime}$.

The correct identification $\Delta W_{k}=-\Delta E_{k}$\ or $\Delta
W_{k}^{\prime}=-\Delta E_{k}^{\prime}$\ ensures that%
\[
\Delta_{\text{i}}W_{k}=\Delta_{\text{i}}W_{k}^{\prime}\geq0
\]
in conformity with its ubiquitous nature as is evident from the existence of
pressure fluctuations in Eq. (\ref{PressureFluctuations}); see Conclusion
\ref{Conclusion-nonzeo-diW}. With this identification, Eqs.
(\ref{GujratiRelation}) and (\ref{GujratiRelation0}) always hold in the
exclusive and inclusive approaches, respectively. They hold not only for free
expansion, but also for any arbitrary process $\mathcal{P}_{0}$ as we show in
Sec. \ref{Sec-FreeExpansion}. Thus, our theorem differs from the JE, which
fails for free expansion.

With the ubiquitous nature of $d_{\text{i}}W_{k}=-d_{\text{i}}E_{k}$ and the
SI-equivalance of $dW_{k}=-dE_{k}$ as opposed to the mistaken identity
$d\widetilde{W}_{k}=dE_{k}$, we believe that the correction requires complete
reassessment of current applications; see Sec. \ref{Sec-MistakenIdentity}. In
particular, we need to recognize the importance of $d_{\text{i}}%
W_{k}=-d_{\text{i}}E_{k}$, which has not been hitherto recognized. Without
acknowledging this fact in any theory results in a total absence of
dissipation regardless of the time dependence the work protocol such as in the
Jarzynski equality.


\begin{thebibliography}{99}
\bibitem {Note-xi-dependence}How the entropy $S$ and the Hamiltonian
$\mathcal{H}$ develops a dependence on $\boldsymbol{\xi}$ is explained by an
example in Sec. \ref{Sec-Friction}.
\end{thebibliography}
\end{document}